\documentclass[journal,comsoc,10pt]{IEEEtran}
\usepackage{booktabs}
\usepackage{subfig}
\usepackage{amsmath}
\usepackage{amsthm}
\usepackage{amsfonts}
\usepackage{tabularx}
\usepackage{url}
\ifCLASSOPTIONcompsoc
  \usepackage[nocompress]{cite}
\else
  \usepackage{cite}
\fi
\usepackage{algorithm}
\usepackage[noend]{algcompatible}
\usepackage[pdftex]{graphicx}
\graphicspath{{../Figures/}}
\DeclareGraphicsExtensions{.pdf,.jpeg,.png,.JPG,.PNG}

\newcommand*{\QED}{\hfill\ensuremath{\square}}

\newtheorem{theorem}{Theorem}
\newtheorem{lemma}{Lemma}
\newtheorem{defn}{Definition}

\usepackage[textsize=tiny]{todonotes}
\begin{document}
\title{Popular Topics Spread Faster: New Dimension for Influence Propagation in Online Social Networks}
 \author{\IEEEauthorblockN{Tianyi~Pan, Alan~Kuhnle~\IEEEmembership{Student~Member,~IEEE}, Xiang~Li~\IEEEmembership{Student~Member,~IEEE} and My~T.~Thai~\IEEEmembership{Senior~Member,~IEEE}}\\
\IEEEauthorblockA{CISE Department,
University of Florida, Gainesville, FL, USA \\ Email: \{tianyi,kuhnle,xixiang,mythai\}@cise.ufl.edu}
 }
\author{Tianyi~Pan,
        Alan~Kuhnle,~\IEEEmembership{Student~Member,~IEEE},
        Xiang~Li,~\IEEEmembership{Student~Member,~IEEE}
        and~My~T.~Thai,~\IEEEmembership{Senior~Member,~IEEE}%
\thanks{T. Pan, A. Kuhnle, X. Li and M. T. Thai are with the Department of Computer \& Information Science \& Engineering, University of Florida, Gainesville,
FL, 32611 USA. Email: \{tianyi, kuhnle,xixiang, mythai\}@cise.ufl.edu }.}%


\maketitle

\begin{abstract}
Information can propagate among Online Social Network (OSN) users at a high speed, which makes the OSNs become important platforms for viral marketing. Although the viral marketing related problems in OSNs have been extensively studied in the past decade, the existing works all assume known propagation rates and are not able to solve the scenario when the rates may dynamically increase for popular topics. In this paper, we propose a novel model, Dynamic Influence Propagation (DIP), which allows propagation rates to change during the diffusion and can be used for describing information propagation in OSNs more realistically. 
Based on DIP, we define a new research problem: Threshold Activation Problem under DIP (TAP-DIP). TAP-DIP is more generalized than TAP and can be used for studying the DIP model. However, it adds another layer of complexity over the already \#P-hard TAP problem. Despite it hardness, we are able to approximate TAP-DIP with $O(\log|V|)$ ratio. Our solution consists of two major parts: 1) the Lipschitz optimization technique and 2) a novel solution to the general version of TAP, the Multi-TAP problem. We experimentally test our solution Using various real OSN datasets, and demonstrate that our solution not only generates high-quality yet much smaller seed sets when being aware of the rate increase, but also is scalable. In addition, considering DIP or not has a significant difference in seed set selection. 
\end{abstract}

\section{Introduction}\label{sc:introduction}
\IEEEPARstart{O}{SNs} have become effective channels for influence propagation as users of OSNs tend to reply/forward the content they are interested in, which makes the content visible to all other users in their social circles. One use case of influence propagation is the viral marketing campaigns \cite{Domingos01,Richardson02}, where companies provide free samples of products to some influential individuals (seed nodes), in order to spread the product information to at least a certain number of users via word-of-mouth effect. Since the seminal paper by Kempe et al. \cite{kempe2003maximizing}, influence propagation in OSNs has been studied in various contexts \cite{Leskovec07,Chen10,Goyal11b,du2013scalable,Cohen14,Borgs14,Tang14,Tang15,nguyen2016targeted,li2017approximate,Nguyen2016stop,long2011minimizing, goyal2013minimizing,nguyen2013least,zhang2014minimizing,dinh2014cost,kuhnle2017scalable}. 

In previous works, there are two major types of influence propagation models: 1) the Triggering Model \cite{kempe2003maximizing} in which the influence propagates in rounds and thus the propagation rate is uniform; 2) the Continuous-Time Diffusion Model \cite{du2013scalable}, in which the propagation rate is decided by a probability density function (pdf) for each edge. The two existing models share a common feature that they are \textit{static}: whether the propagation rate is constant or follows a pdf, it is known before the propagation starts. 

Unfortunately, the models are not able to characterize one important property: the propagation rate may change \textit{during} the propagation. As we confirmed by analyzing the retweet data we crawled using Twitter API, a topic is likely to propagate at a faster rate once it becomes popular (trending). Thus, to better depict influence propagation in reality, it is necessary to develop a propagation model that enables change of propagation rates \textit{based on the current propagation status}. 

Therefore, we propose the DIP model which can explicitly consider the rate change. In the model, we follow the idea in literature \cite{zhang2014minimizing} that a topic becomes popular when number of influenced nodes surpasses a predefined amount. Notice that the condition may not always reflect the complicated conditions in reality (e.g. Twitter has internal algorithms for determining trending topics), but it is a reasonable abstraction. We then formulate the problem TAP-DIP To analytically study the model. TAP-DIP asks for a seed set with minimum size that guarantees that the number of nodes influenced can reach a certain threshold within time limit, when the propagation model is defined under DIP. 

The main challenge of TAP-DIP is resulted from its one new dimension, dynamic propagation rates, as it creates an obstacle for using the sampling techniques that are widely applied to solve influence propagation related problems. The sample techniques are important for influence propagation since even computing the exact influence is \#P-hard \cite{Chen10}. Each sample provides information on what nodes can be influenced by a certain node (forward sampling \cite{kempe2003maximizing}) or the set of nodes that may influence a target node (reverse sampling \cite{Borgs14}). Obviously, the samples have no access to global information such as total number of increased nodes and propagation rate change. Thus, the sampling methods cannot be easily adapted for solving TAP-DIP.

To tackle the challenges brought by dynamic propagation rates, we propose the algorithm FAST (stands for Finding Anticipated Speedup Time) which can decide the near-optimal time that the propagation rate may increase, in terms of minimizing the number of seeds used to trigger the increase together with those to guarantee reaching thresholds. FAST breaks down TAP-DIP into subproblems that the rate increases happen at fixed times, in which case the sampling methods are again applicable. However, the subproblem is still complicated as it needs to meet the thresholds for both triggering the rate increase and satisfying the activation requirement. To solve the subproblem, we designed the first efficient algorithms for both Multi-TAP (MTAP) and Multi-Influence Maximization (MIM). FAST can solve the TAP-DIP problem with approximation ratio $2\log|V|$ ($V$ is the set of nodes in the OSN), which is close to the best ratio $\log|V|$ that one can expect for TAP \textit{without} DIP. We run extensive experiments to demonstrate the efficiency of FAST and explore various settings of the DIP model using several real-world OSN datasets.

In summary, our contributions are as follows. 
\begin{itemize}
	\item We propose TAP-DIP, the first influence propagation related problem in OSNs that explicitly considers propagation rate increase. We support the validity of the model by data analysis results from crawled retweets in around 4,000 Twitter trending topics.  
	\item We propose the algorithm FAST to solve TAP-DIP. It is the first solution to TAP-DIP with approximation ratio of $2\log |V|$. The two subroutines of FAST, MMinSeed and Multi-IM, are the first algorithms that can efficiently solve the MTAP problem and the MIM problem, respectively.
    \item We perform extensive experiments on various real OSN data sets to demonstrate both the efficiency of our proposed algorithms and the drastic difference in the solutions when considering rate increase. 
\end{itemize}

\textbf{Related Work.}
Kempe et al. \cite{kempe2003maximizing} are the first to study influence propagation in OSNs mathematically. Their focus was on the Influence Maximization problem (IM), which drew much attention in the research community \cite{Leskovec07,Chen10,Goyal11b,du2013scalable,Cohen14,Borgs14,Tang14,Tang15,nguyen2016targeted,Nguyen2016stop,li2017approximate}. Another major problem is TAP \cite{long2011minimizing,goyal2013minimizing,nguyen2013least,zhang2014minimizing,dinh2014cost,kuhnle2017scalable}. The main propagation model adopted in the papers is the Triggering model \cite{kempe2003maximizing} or its variations, the Independent Cascading (IC) model or the Linear Threshold (LT) model. Another model that considers variation in propagation rate is the continuous time diffusion model \cite{du2013scalable}. However, both models assume known and fixed parameters for the diffusion, which may not represent the real-world scenarios. 

As even computing the exact influence is \#P hard \cite{Chen10}, the mainstream approach of solving IM or TAP relies extensively on sampling, which is inefficient (due to many redundant samples) until Borgs et al. proposed the Reverse Influence Sampling (RIS) method in \cite{Borgs14}. The RIS method was further refined in \cite{Tang14,Tang15,nguyen2016targeted,Nguyen2016stop} for better time complexity. However, the RIS method was not yet applied to solve TAP or MIM, nor can it consider the DIP model. The only exception is \cite{pan2017dynamic}, from which this paper is extended. 

\textbf{Organization.} 
The rest of the paper is organized as follows. In Section \ref{sc:model}, we present our analysis on propagation rates, describe the DIP model and the TAP-DIP problem. Section \ref{sc:los} and \ref{sc:mminseed} discuss our solution, FAST to TAP-DIP. The performance of FAST and the behavior of the DIP model are analyzed in Section \ref{sc:exp}. Section \ref{sc:conclusion} concludes the paper and provides some insights on how FAST can be extended to solve a more generalized TAP-DIP problem.

\section{Model and Problem Definition}\label{sc:model}
In this section, we first analyze the Twitter data we crawled, which provides solid evidence that being trending will highly likely increase the propagation rate of a topic. Then we introduce the graph model of the OSN and the dynamic influence propagation model. We present the formal definition of TAP-DIP at the end of this section.  

\subsection{Analysis of Twitter Data}\label{ssc:dataanalysis}
We crawled the tweet stream data for 5,049 different Twitter trending topics in the US using the REST APIs\footnote{\url{https://dev.twitter.com/rest/public}}, during the period of Nov. 2016 to Apr. 2017. Specially, we collected the \textit{retweets} whose times are within three days of the time that the topic \textit{first} became trending. In order to decide the trending times for the topics, we first maintain the collection of all trending topics in three days and then crawl the current trending topics every 5 minutes. The trending time is considered as the first time that a new trending topic is recorded. We also update the collection when necessary. When the trending time of a topic is decided, we can use the Search API in REST\footnote{\url{https://dev.twitter.com/rest/public/search}} to fetch the historical retweets within the desired times. The retweets are separated by the trending time and into two groups, before trending and after trending, as our major goal is to demonstrate that the propagation rate, which is the reciprocal of the time difference between the retweet and the original tweet (retweet delay) in this case, increases after the topic being trending. We omit the topics having less than 100 retweets before/after trending to avoid outliers such as promoted trends and we are left with 3,988 topics after this step. For each remaining topic, we calculated the time difference for all its retweets. Based on the arrays of time differences before/after trending, we can decide whether the time difference decreased (or equivalently, propagation rate increased) after trending, using KS-test \cite{degroot2011kolmogorov} and t-test.

\begin{figure}[!ht]
\centering
		\includegraphics[width = 0.75\linewidth]{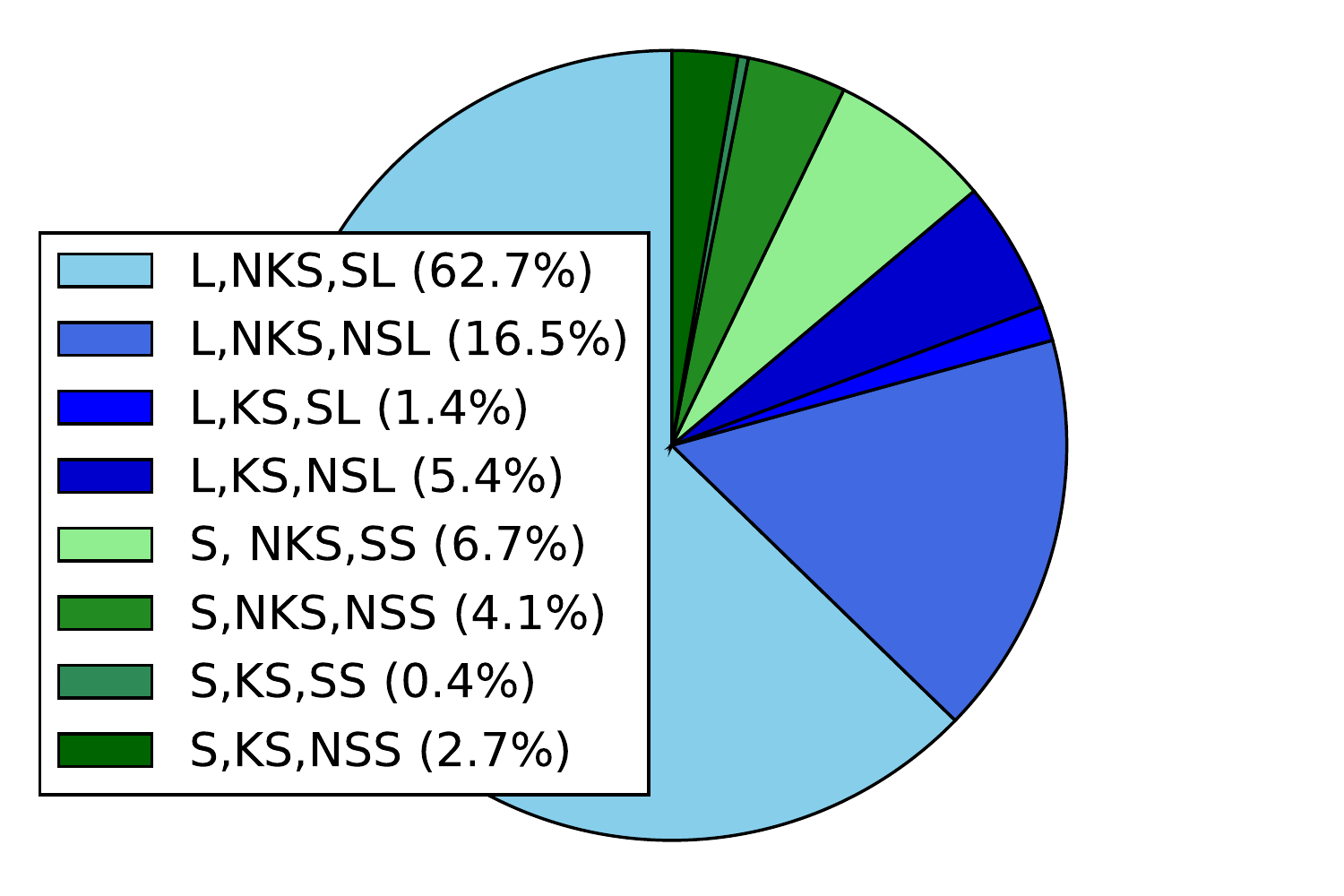}		
        \caption{Statistical Test Results. L/S means the average propagation rate after trending is larger/smaller than before. Then SL/NSL or SS/NSS denotes if it is significantly large/small by t-test. NKS/KS means reject or cannot reject the null hypothesis of KS test.}\label{fig:testres}
\end{figure}

\begin{figure}[!ht]
\centering
		\includegraphics[width = 0.75\linewidth]{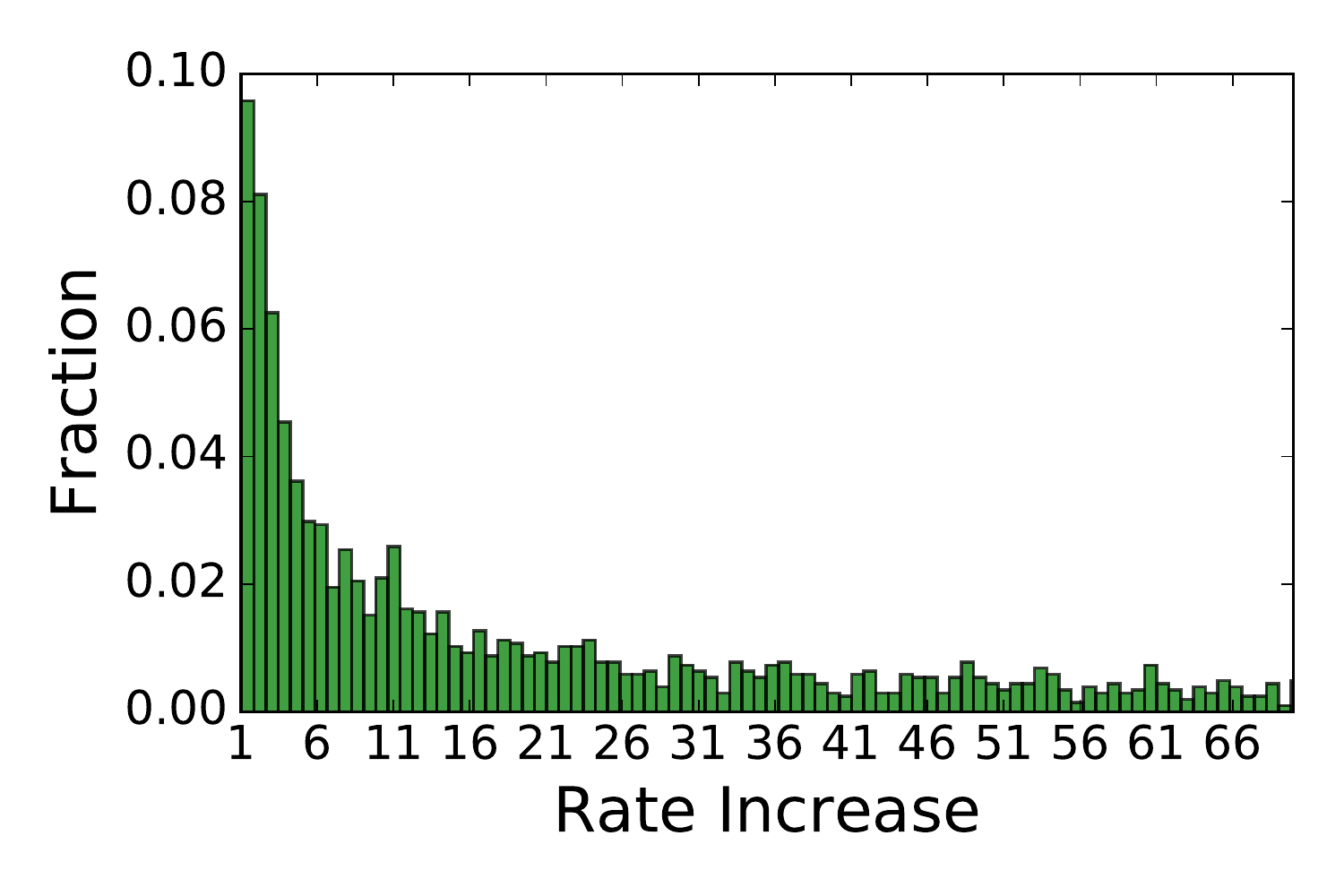}
        \caption{Distribution of Rate Increase}				\label{fig:quotient}
\end{figure}

In Fig.~\ref{fig:testres}, we present the test results of all the 3,988 topics. It is clear that increment of propagation rate after topic being trending is a common phenomenon, as $86.1\%$ of the topics have shorter average time delay for retweeting. Among those topics (with label ``L''), $72.8\%$ have significantly increase in propagation rate (``L,NKS,SL''), confirmed by both KS-test and t-test. 

We further study the distribution of propagation rate increase of those increased significantly (2,502 topics in total). Define the propagation rate increase as $\frac{\text{avg. retweet delay before trending}}{\text{avg. retweet delay after trending}}$, we obtain the histogram of rate increase in Fig.~\ref{fig:quotient}, where the y-axis denote the fraction of topics having rate increase in the range. The rate increase distribution is heavy tailed with some concentration on small numbers (1-10). A closer inspection reveals that higher rate increases are likely due to special events. For example, the topic ``Tim Duncan'' was trending on Dec. 19, 2016, the day right after the former NBA star's jersey retirement. The average retweet delay before trending was about 4 days and for after trending, only 4 hours, which leads to a 24 times rate increase. In a more extreme case, the topic ``President Trump" was trending on the inauguration day and the propagation rate increased 100 times. For those topics, being trending may only have minor impact on their popularity. For viral marketing related topics, the rate increases are usually moderate. For instance, the topics ``And iPhone'' (product), ``Cuisinart'' (brand) had $6$ and $8$ times increase in rate after being trending, respectively.

\subsection{The DIP Model}\label{ssc:dip}
We abstract the OSN to be a directed, connected graph $G=(V,E)$, where $V$ denotes all the users in the OSN, and $E$ corresponds to the relationships among the users (follow, friend, etc.) Each edge $(u,v)\in E$ is associated with a weight $p_{uv}\in [0,1]$ and a probability density function $l_{uv}(\tau)$, which are used to characterize the influence propagation model that is detailed in the following. Also, we consider the propagation rate at time $t$ as $\rho(t)$, which is defaulted at $1$ and $\rho(t)>1$ means a faster propagation. 

To model the change in influence propagation rate while considering the impact of social relationship strength, we combine the IC model and the Continuous-Time Diffusion Model into the Continuous IC Model (CIC), whose definition is as follows. We denote the initial set of activated (influenced) nodes as $S$. 
\begin{defn}[Continuous IC Model]
Consider a graph $G=(V,E)$ with $l_{uv}(\tau)$ and $p_{uv}$ defined on each edge $(u,v)\in E$. The influence diffusion process starts when all nodes in $S$ are activated at time $t=0$ and all other nodes remain unactivated. When node $u$ is activated at time $t$, each neighbor $v$ of $u$ will be activated at time $t+\tau/\rho$ with probability $p_{uv}$ where $\tau$ follows the probability density function $l_{uv}(\tau)$. Once a node is activated, it will never be deactivated. The process stops when no more nodes can be activated. 
\end{defn}
Based on the CIC model, a trending topic can propagate faster, as it has a shorter delay in transmission. We denote $I_{t}(S)$ as the total number of nodes influenced at time $t$, given the initial seed set $S$. Also, we write $\mathbb{I}_{T}(S)=E[I_{t}(S)]$ as the \textit{influence spread} of seed set $S$ up to time $t$. 

With the definition of $I_{t}(S)$ and the connection between trending and faster propagation, we can characterize propagation rate $\rho$ as a function of time $t$. Based on the idea of \cite{zhang2014minimizing}, a topic is popular when the number of influenced nodes reaches a threshold. In this paper, we consider a topic becomes trending at time $t$ when the number of nodes influenced by the seed set $S$ is larger than fraction $\phi$ of nodes in $V$. When the seed set $S\subseteq V$ is known and rate parameter $r>1$, we have:
\[
  \rho(t)=\begin{cases}
               1, I_{t}(S)< \phi |V|\\
               r, I_{t}(S)> \phi |V| 
            \end{cases}
\]
In this section and for the major part of the paper, we focus on the case when the propagation rate can change only once, as such a case corresponds to our findings from data analysis. We will discuss the possibility of having multiple propagation rate changes in Section \ref{sc:conclusion}.

\subsection{The TAP-DIP Problem.}
In a viral marketing campaign, the goal can often be influencing at least a certain number of users within a period of time. For example, a company showcasing its new product will want it to be exposed to a certain percentage of the market within a few days after the release. The company needs to choose some users as seeds to propagate the product information, and seeding each user incurs a cost. For cost-effectiveness, companies always want to minimize the number of seed users, when the costs of seeding the users are the same. Thus, the problem can be rephrased as finding a seed set with minimum size such that the number of activated nodes can be at least a certain threshold $\eta |V|$. Such a problem is termed as the \textit{Threshold Activation Problem (TAP)}. In a typical TAP problem, the underlying influence propagation model is often static, that the parameters of the model will not change overtime. TAP-DIP, however, is the version of TAP that considers dynamic influence propagation models. In this paper, we focus on the TAP problem with the propagation model CIC. 

\begin{defn}[TAP-DIP]
Given an OSN $G=(V,E)$ with $l_{uv}(\tau)$ and $p_{uv}$ defined on each edge $(u,v)\in E$, the activation threshold $\eta$, the trending triggering threshold $\phi$, the propagation rate $r$, the time limit $T$, TAP-DIP asks to find a seed set $S$ with minimum size such that the influence spread $\mathbb{I}_{T}(S)$ is at least $\eta|V|$ within time $T$.
\end{defn}

In the following two sections, we propose FAST, our solution to TAP-DIP. For conciseness, most of the proofs are placed in the appendix.
\section{FAST: Solution to TAP-DIP}\label{sc:los}
\subsection{Overview}\label{ssc:losoverview}
For all existing solutions to influence propagation related problems, a known propagation model is required, which is not possible in TAP-DIP as the propagation rate may change based on number of influenced nodes. It seems that we have to derive brand new solutions for TAP-DIP, however, it is not the case when we can provide a key value: the time $t$ that the propagation rate changes. In TAP-DIP, this value is a variable based on number of influenced nodes. When added as an input, it defines fixed propagation models, yet it also brings in the constraint that the number of influenced nodes must meet the triggering threshold $\phi$ at time $t$. Thus, with a fixed $t$, TAP-DIP can be reduced to a problem of finding the minimum seed set to reach $\phi|V|$ and $\eta|V|$ thresholds at time $t$ and $T$, respectively. Notice that the propagation models in times $[0,t]$ and $(t,T]$ are different due to the change in propagation rate. We term the problem as Multi-TAP (MTAP), which is a generalization of the TAP problem that only considers satisfying a single threshold $\eta|V|$ at time $T$, but is still much more accessible than TAP-DIP itself. 

As the actual solution to MTAP is complicated, we delay its details in Sect.~\ref{sc:mminseed} and assume for now that it is available in a blackbox. We can feed a $t$ value to it and obtain a seed set $S_t$. Clearly, the solution to TAP-DIP is the $S_t$ with minimum cardinality. However, the function $H(t) = |S_t|$ has no closed form and we have to use global optimization techniques to find its minimum. At the core of our FAST algorithm is the Lipschitz optimization \cite{horst2013handbook} framework, which can find globally near-optimal values of a function given that the function is Lipschitz continuous, with limited calls to function value calculation. 

\begin{defn}[Lipschitz Continuity]
A function $f(x)$ defined on $X$ is Lipschitz Continuous if there exists a real constant $L\geq 0$, such that for all $x_1,x_2\in X$, 
\begin{align}
	|f(x_1)-f(x_2)|\leq L|x_1-x_2|\label{eq:lipschitz}
\end{align}
\end{defn}

In the following, we first prove that an approximation $H'(t)$ of the function $H(t)$ is Lipschitz continuous, and then propose the algorithm FAST that finds near-optimal values of $H(t)$ over $t$.

\subsection{The FAST algorithm}\label{ssc:losdetail}

If we denote $\Delta$ as the minimum distance between any two possible values of $t$, we can easily prove that $H(t)$ is Lipschitz continuous with $L=|V|/\Delta$. Unfortunately, Lipschitz optimization can hardly benefit from such a crude estimation of $L$. The larger the $L$ value, the more complicated the problem. Therefore, we introduce a relaxed version of $H(t)$ that a much smaller $L$ is achievable. 
 
We write the relaxed version of $H(t)$ as $H'(t)=S^*_s(t)+S^*_a(t)$, where $S^*_s(t)$ denotes the minimum number of nodes to guarantee $\phi$ fraction of influenced nodes in $G$ and therefore a speed-up at $t$; $S^*_a(t)$ denotes the minimum number of nodes to guarantee $\eta$ fraction of activation in $G$ \textit{given} a speed-up at $t$. Notice that $S^*_a(t)$ is calculated based on the assumption that the nodes triggered the speed-up did not influence any nodes in $G$. Therefore, $H'(t)$ is an upper bound on the number of required nodes, as stated in the following lemma.

\begin{lemma}\label{lemma:ddipsinglebound}
$H'(t^{*'})\leq 2H(t^*)$ where $H'(t^{*'})\leq H'(t), H(t^*)\leq H(t), \forall t\in [0,T]$.
\end{lemma}

It is clear that $S^*_s(t),S^*_a(t)$ are monotonically decreasing/increasing with  $t$, respectively. Such properties lead to a refined result on the Lipschitz continuity of $H'(t)$.

\begin{lemma}\label{lemma:locallipschitzconstant}
Given an interval $[t_1,t_2]$ and function values $H'(t_1)=S^*_s(t_1)+S^*_a(t_1), H'(t_2)=S^*_s(t_2)+S^*_a(t_2)$, $H'(t)$ is Lipschitz continuous over $[t_1,t_2]$ with constant 
\begin{align}
l_{t_2}=\frac{1}{\Delta}\max \{S^*_s(t_1)-S^*_s(t_2),S^*_a(t_2)-S^*_a(t_1)\}\label{eqn:locallipschitz}
\end{align}
\end{lemma}

The following lemma adapted from \cite{lera2013acceleration} ensures a lower bound on function values within any interval of $H'(t)$.

\begin{lemma}\label{lemma:locallowerbound}
When $H'(t)$ is Lipschitz continuous over $[t_1,t_2]$ with constant $L$,
\begin{align}
	\min_{t\in [t_1,t_2]}H(t)\geq \frac{H'(t_1)+H'(t_2)}{2}-\frac{l_{t_2}(t_2-t_1)}{2}\label{eqn:locallowerbound}
\end{align}
\end{lemma}

With Lemmas \ref{lemma:locallipschitzconstant} and \ref{lemma:locallowerbound}, we propose FAST for finding the global minimum of $H'(t)$ over $[0,T]$, assuming that we have access to the values of $H(t)$ and $H'(t)$, which will be detailed in Sect.~\ref{sc:mminseed}.
\begin{algorithm}
	\caption{Finding Anticipated Speedup Time (FAST)}
	\label{alg:lipschitz}
    \begin{algorithmic}
    \REQUIRE $H'(t), H(t), t\in[0,T]$
	\ENSURE The global minimizer $\bar{t}$ of $H'(t)$, $H(\bar{t})$
    \STATE Calculate $H'(0), H'(T)$ by Alg. \ref{alg:mminseed}.
    \STATE Let $t_1=0,t_{2}=T$, $i,k=2$, Calculate $l_{i}$ based on \eqref{eqn:locallipschitz}
	\WHILE{$|t_i-t_{i-1}|\geq \frac{1}{l_{i}}$}
      \STATE $t^{k+1}=\frac{t_i+t_{i-1}}{2}+\frac{H'(t_{i-1})-H'(t_i)}{2l_i}$, $k$++
      \STATE Calculate $H'(t^{k+1})$ by Alg. \ref{alg:mminseed}
      \STATE Renumber all points such that $0\leq t_1\leq \cdots\leq t_k\leq T$
      \FOR{Each interval $[t_{j-1},t_j]$}
        \STATE Calculate $l_j$ based on \eqref{eqn:locallipschitz} and $R_j$ based on rhs of \eqref{eqn:locallowerbound}
      \ENDFOR
      \STATE Let $i=\arg\min_{j=1,\cdots,k}\{R_j\}$
	\ENDWHILE
    \STATE $H'(\bar{t})=\min \{H'(t_i)|i=1,\cdots,k\}$
    \STATE $\bar{t} = \arg\min\{H'(t_i)|i=1,\cdots,k\}$
    \STATE Calculate $H(\bar{t})$ by Alg. \ref{alg:mminseed}.
	\end{algorithmic}
\end{algorithm}

FAST utilizes the Lipschitz continuity of $H'(t)$. Intuitively, it iteratively finds the interval with minimum lower bound by Lemma~\ref{lemma:locallowerbound} and calculate a new value in the interval, until the interval is small enough. By Lipschitz continuity, the minimum of $H'(t)$ is close to one of the calculated values. We further refine the result by calculating $H(\bar{t})$ instead of $H'(\bar{t})$ at the end. The following theorem guarantees the solution quality.
\begin{theorem} \label{theorem:lipschitzratio}
$H(\bar{t})\leq 2H(t^*)+1$.
\end{theorem}


\section{MMinSeed: Solution to Multi-TAP}\label{sc:mminseed}
In this section, we describe the missing piece in Section \ref{sc:los}: how to calculate $H'(t),H(t)$, which completes the full picture of FAST. As discussed in Sect.~\ref{ssc:losoverview}, calculating $H(t)$ is actually solving MTAP with two thresholds. For $H'(t)$, its two components $S^*_s(t),S^*_a(t)$ can be seen as TAP instances with threshold $\theta|V|$ and time limit $t$, threshold $\eta|V|$ and time limit $T$, respectively. Therefore, a solution to MTAP suffices for completing FAST. In the following, we propose the first efficient solution to a generalized version of MTAP (defined below) that considers multiple thresholds and time limits. Thus, our proposed solution MMinSeed is not only capable for solving $H(t),H'(t)$, but also applicable to the general scenarios. 
\begin{defn}[MTAP]
Given $G(V,E)$, $V^l\subseteq V, l=1,...,L$, thresholds $\eta_1,\cdots,\eta_L$, time limits $t^1,...,t^L$ and the propagation model, MTAP asks for a seed set $S$ that can influence, by expectation, $\eta_l|V^l|$ nodes in each subset within $t^l$.
\end{defn}
In MTAP, we assume that a trending threshold is met when the expected influence, but not the actual influence, is larger than the threshold as it is not possible to obtain the actual influence when calculating the seed set. Although it may not be exactly the same as in DIP, this assumption is still acceptable: the influence is usually concentrated at the expectation \cite{zhang2014minimizing}. We also demonstrate in our experiments that the performance of the algorithms are satisfactory, in which we run extensive simulations to demonstrate the performance when we trigger speedups by actual influence.




\subsection{The RIS Framework}
Due to the complexity from the probabilistic network, sampling is the most popular method to estimate the influence spread of a seed set in each ground set. Here we adopt the state-of-art Reverse Influence Sampling (RIS) technique \cite{Borgs14} for generating samples. Specially, we combine the sampling methods in two recent papers \cite{Tang15,nguyen2016targeted} to generate samples for each ground set under Continuous IC propagation model.

The RIS approach has two phases. In the sample generation phase, a number of samples are generated, where each sample consists of all nodes that can influence a random node in a realization of the probabilistic graph. In the seed set selection phase, a maximum coverage problem (with nodes as sets and samples as elements) is greedily solved to obtain the seed set for influence maximization. 

\subsection{The MMinSeed algorithm}
To adapt the RIS framework to solve MTAP, there are two obstacles. The first one is how to guide the solution to consider all the thresholds at the same time. A second and more challenging one is that, the RIS framework is designed for maximizing influence with a fixed number of seeds. Also, the number of samples required to guarantee a certain level of accuracy will increase with more seed nodes. However, MTAP asks for minimizing the seed set size, which is unknown and cannot be used to determine the number of required samples. 

In order to overcome the first obstacle, we need to design an objective function that satisfies the following conditions: (1) Maximizing the function will fulfill all thresholds. (2) When a threshold is fulfilled, additional influence to the corresponding ground set should not bring any benefit to the function (3) The function must be submodular. The first two conditions insure the correctness of the function, while the last one guarantees the performance, as otherwise no approximation ratio will exist. 

Based on the conditions, we design the function $f(S)$, which is defined as
\begin{align}
f(S)=\sum_{l\in L}\min\{\eta_l|V^{l}|, |V^l(S)|\}, \quad S\subseteq \mathcal{S}\label{eqn:defncombined}
\end{align}
where $V^l(S)$ denotes the nodes influenced by $S$ in the ground set $V^l$. Clearly, $f(S)$ is submodular and monotone increasing as it is the summation of submodular functions. Also, each ground set can contribute up to its threshold to the function value. Additionally, the maximum of this function can only be achieved when all thresholds are fulfilled.  


We describe MMinSeed in Alg.~\ref{alg:mminseed}. In MMinSeed, the process of finding the number of seeds utilizes the submodularity of $f(.)$. In each round, MMinSeed calculates the average gain in $f(.)$ of a seed node added in the previous round and the gap from the current $f$ value to the requirement. Then it decides how many new seed nodes are required, assuming all the new nodes can bring the gain equal to the average value calculated. The approach will reduce the number of calls to its subroutine, Alg. \ref{alg:multiim}. Comparing with binary search, the greatest advantage of this approach is that it will never choose a seed set size that is larger than necessary, which is guaranteed by submodularity. A larger seed set is not preferable since it leads to generating more samples, which is redundant and costs extra time.

\begin{algorithm}
	\caption{MMinSeed}
    \label{alg:mminseed}
    \begin{algorithmic}
      \REQUIRE Graph $G=(V,E)$, Ground sets $V^1,\cdots,V^L$ with thresholds $\eta_1,\cdots,\eta_L$, $\epsilon>0$
      \ENSURE Seed set $S\subseteq V$
      \STATE $f(S^*)=\sum_{l\in L}\eta_l|V^{l}|$.
      \STATE $j=1, j_{prev}=0$, $f=0, f_{prev}=0$, $S=\emptyset$
      \STATE Find $S$ using Alg. \ref{alg:multiim} with $|S|\leq j$, $f=\hat{f}(S)$
      \WHILE{$\hat{f}(S)<(1-\epsilon)f(S^*)$}
          \STATE $j += \lceil\frac{(1-\epsilon)f(S^*)-\hat{f}(S)}{(f-f_{prev})/(j-j_{prev})}\rceil$
          \STATE Find $S$ using Alg. \ref{alg:multiim} with $|S|\leq j$
      \ENDWHILE
    \end{algorithmic}
\end{algorithm}

The subroutine Multi-IM (Alg.~\ref{alg:multiim}), is the first efficient solution to the MIM problem. It maintains $L$ collections of samples $\mathcal{R}^1,\cdots,\mathcal{R}^l$ (one for each threshold) and it keeps generating new samples for each collection up to a given amount $N_{\mathcal{R}^l}$. Then, Alg. \ref{alg:multiim} greedily solves a submodular maximization problem with the submodular function $f(S)$ defined in \eqref{eqn:defncombined}, using at most $k$ nodes. The resulting set $S_k$ is used to verify if the number of samples intersect with $S_k$, $C_{\mathcal{R}^l}(S_k)$, is at least $\gamma$. If the verification is successful, $\mathcal{R}^l$ is enough to guarantee the accuracy of estimating the ground set. It then stops generating new samples for $\mathcal{R}^l$. Otherwise, it doubles $N_{\mathcal{R}^l}$, generating samples up to $N_{\mathcal{R}^l}$ and rerun the verification. When all $\mathcal{R}^l$ passed the verification, the solution $S_k$ is returned as output. 

\begin{algorithm}
	\caption{Multi-IM}
    \label{alg:multiim}
    \begin{algorithmic}
    \REQUIRE Graph $G=(V,E)$, Ground sets $V^1,\cdots,V^L$ with thresholds $\eta_1,\cdots,\eta_L$, Precision parameters $\epsilon>0$, $\delta\in(0,1)$
	\ENSURE Seed set $S\subseteq V$
   	\STATE Collection of samples $\mathcal{R}^l = \emptyset, l=1,\cdots,L$
    \STATE $\phi = \frac{(1-1/e)\sigma+\tau}{\epsilon}$, $\gamma=2(\phi^2 + \log \frac{3L^2}{(2L-1)\delta})$ 
    \STATE $N_{\mathcal{R}^l}=\gamma, ctn^l=true, l=1,\cdots,L$
	\WHILE{$\exists l$ s.t. $ctn^l == false$}
        \STATE $S_k$ = Greedy size $k$ solution to maximize $\hat{f}(S)$
    	\FOR{Each $l=1,\cdots,L$}
        	\IF{$ctn^l$}
                \IF{$C_{\mathcal{R}^l}(S_k)\geq\gamma$}
                	\STATE $ctn^l=false$
                \ELSE         	
                	\STATE Generate $N_{\mathcal{R}^l}$ samples for $\mathcal{R}^l$, $N_{\mathcal{R}^l}=2N_{\mathcal{R}^l}$
                \ENDIF
            \ENDIF
        \ENDFOR
    \ENDWHILE
    \end{algorithmic}
\end{algorithm}

\subsection{Theoretical Analysis.}\label{ssc:theoreticalanalysis}
Since $f(.)$ is submodular, the following result \cite{elomaa2010covering} holds using the greedy algorithm, if all $f(.)$ values can be obtained in polynomial time:
\begin{align*}
f(S_j^g)\geq (1-(1-1/k)^j)f(S_k^*)
\end{align*}
where $S_j^g$ is the collection of the first $j$ sets selected by the greedy algorithm to maximize $f(.)$ and $S_k^*$ is the optimal collection of size $k$. However, in Multi-IM, the values of $f(.)$ are not accurate but estimated by RIS. Hence, we can only have a weaker result as in Theorem~\ref{theorem:greedyguarantee}. 
\begin{theorem}\label{theorem:greedyguarantee}
Alg. \ref{alg:multiim} guarantees
\begin{align}
f(S_j^g)\geq (1-(1-1/k)^j-\epsilon)f(S_k^*) \label{eqn:greedyguarantee}
\end{align}
with probability at least $1-\delta$. 
\end{theorem}

To prove Theorem~\ref{theorem:greedyguarantee}, we prove the following two Lemmas. In lemma~\ref{lemma:numberofsamples}, we derive the number of samples required to guarantee \eqref{eqn:greedyguarantee}. Next, we ensure that Alg.~\ref{alg:multiim} generates at least that many samples in Lemma~\ref{lemma:numberofsampleguarantee}. The validity of Theorem~\ref{theorem:greedyguarantee} is then straightforward when combining Lemma~\ref{lemma:numberofsamples} and Lemma~\ref{lemma:numberofsampleguarantee}, and applying the union bound.


\begin{lemma}\label{lemma:numberofsamples}
The number of samples required to guarantee \eqref{eqn:greedyguarantee} with probability at least $1-\frac{L+1}{3L}\delta$ is 
\begin{align}
Q=\sum_{l=1}^LQ^l\label{eqn:totalnumberofsamples}
\end{align}
where 
\begin{align*}
Q^l = \frac{2|V^l|\phi^2}{\mathbb{I}_T^l(S_k^*)\epsilon^2},\sigma =\sqrt{\ln(\frac{3L}{\delta})},\tau = \sqrt{(1-\frac{1}{e})(\ln\frac{3L\binom{|V|}{j}}{\delta})}
\end{align*}
$\delta\in(0,1)$ and $\epsilon>0$ are constants. 
\end{lemma}

\begin{lemma}\label{lemma:numberofsampleguarantee}
Alg. \ref{alg:multiim} guarantees the number of samples for each threshold is at least $Q^l$ when it stops, with probability at least $1-\frac{2L-1}{3L^2}\delta$.
\end{lemma}
For small values of $k$, $\mathbb{I}^l_T(S_k^*)$ can be $0$ for some $l$. In such cases, the above lemmas hold trivially and the number of required samples is defined as $0$. 

With Theorem~\ref{theorem:greedyguarantee}, we are able to derive the approximation ratio of MMinSeed in Theorem~\ref{theorem:greedybicriteriaratio} and eventually, the approximation ratio of FAST in Theorem~\ref{theorem:losratio}.


\begin{theorem}\label{theorem:greedybicriteriaratio}
MMinSeed has approximation ratio $\log |V|$ and achieves at least $(1-\epsilon)$ of the required $f(.)$ value, given \eqref{eqn:greedyguarantee}.
\end{theorem}

\begin{theorem}\label{theorem:losratio}
FAST has approximation ratio of $2\log|V|$.
\end{theorem}
\begin{proof}
The result follows directly by combining Lemma \ref{lemma:ddipsinglebound}, Theorem \ref{theorem:lipschitzratio} and Theorem \ref{theorem:greedybicriteriaratio}.
\end{proof}

\section{Experiments}\label{sc:exp}
\subsection{Experimental Settings.}
The experiments are conducted on a Linux machine with 2.3GHz Xeon 18 core processor and 256GB of RAM. We carry experiments under Continuous IC models on the following datasets from \cite{snap}. \footnote{The source code is available at \url{https://github.com/tianyipan0411/DIP}}

\setlength\tabcolsep{2pt}
\begin{table}[!ht]
	\caption{Datasets' Statistics}
	\label{tab:data_sum}
	\centering
	\begin{tabular}{ l  r  r  r r }\toprule
		\textbf{Dataset} & \bf \#Nodes& \bf \#Edges&\bf T-Node &\bf A-Node \\\midrule
		Facebook & 4K & 88K & 100 - 500 & 1K - 2K\\
        wiki-Vote & 7K & 206K & 100 - 500 & 1K - 2K\\
        Epinions & 76K & 1M & 500 - 2.5K & 10K - 20K\\
        Slashdot & 77K & 1.8M & 500 - 2.5K & 10K - 20K \\
        Twitter & 81K & 3.54M & 500 - 2.5K & 10K - 20K\\
		Gplus& 108K & 26M & 500 - 2.5K & 10K - 20K\\
        Pokec & 1.63M & 61.2M & 5K-20K & 100K-200K\\
		LiveJournal & 4.85M & 138M & 5K-20K & 100K-200K\\
		\bottomrule
		\hline
	\end{tabular}
\end{table}

\noindent\textbf{Datasets.}
We select a set of 8 OSN datasets of various sizes to fully test the impact of the dynamic influence propagation model. The description summary of those datasets is shown in Table \ref{tab:data_sum}.

\noindent\textbf{Parameter Settings.} We follow the papers \cite{Tang14,Goyal11b} for setting propagation probability $p_{uv}$, which is calculated as $p_{uv} = \frac{1}{d^{in}_v}$ where $d^{in}_v$ denotes the indegree of node $v$. We model the propagation rate using Weibull distribution as \cite{du2013scalable,Tang15} and fix the shape parameter at $4$, scale parameter at $1$ throughout the experiments.

In all the experiments, we keep $\epsilon = 0.1$ and $\delta = 1/n$ if the values are not stated otherwise. The time limit is set at 10. The values of $r$, the propagation rate change, varies from $1.5$ to $4.0$. The $\phi,\eta$ values are not set explicitly, instead, we set the number of nodes required for being trending (T-Node) and for overall activation requirement(A-Node), based on the size of the networks. The parameters are summarized in Table \ref{tab:data_sum}. 

\subsection{Performance of FAST.}

Since the TAP-DIP problem is new, there are no suitable algorithms that can be compared directly with FAST. Instead, we demonstrate the performance of FAST via its subroutine. We compare the MMinSeed algorithm with its variation that uses the IMM algorithm \cite{Tang15} instead of the Multi-IM algorithm (Alg.~\ref{alg:multiim}) as a subroutine. In the comparison, we only allow one threshold and modify IMM's sampling method to allow the Continuous-IC model. The comparison is based on the scenario with no rate change and the lowest activation threshold for each network. 

Figure \ref{fig:time_compare} proves a clear difference of the running time.The Multi-IM supported MMinSeed is much faster (the running time is in log scale) than the one supported by IMM. The main reason for the superior performance of Multi-IM is that it decides the sample requirement dynamically, while IMM has a parameter estimation stage to estimate the number of required samples, which can be inaccurate and results in a sample collection that is much larger than necessary when the seed set size is small. In the running time comparison, we set $\epsilon = 0.5$ to allow IMM finish in reasonable time. Each number is the average of 10 runs, as the variation of running time is small and the difference is apparent.
\begin{figure}
\centering
\includegraphics[width=0.75\linewidth]{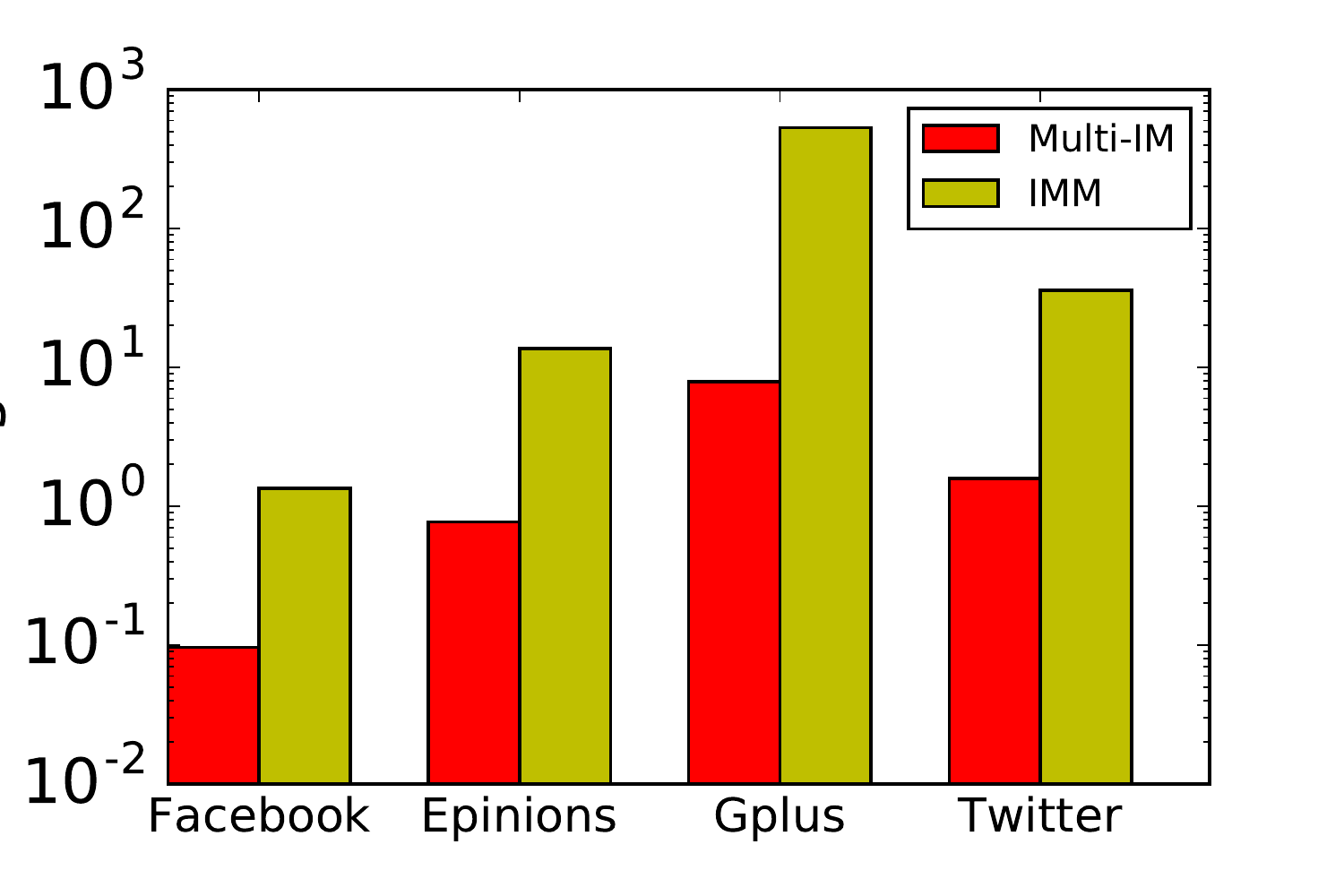}
 	\caption{Efficiency of MMinSeed with Multi-IM/IMM}
 	\label{fig:time_compare}
\end{figure}

\begin{figure}

\centering
 	\includegraphics[width=0.75\linewidth]{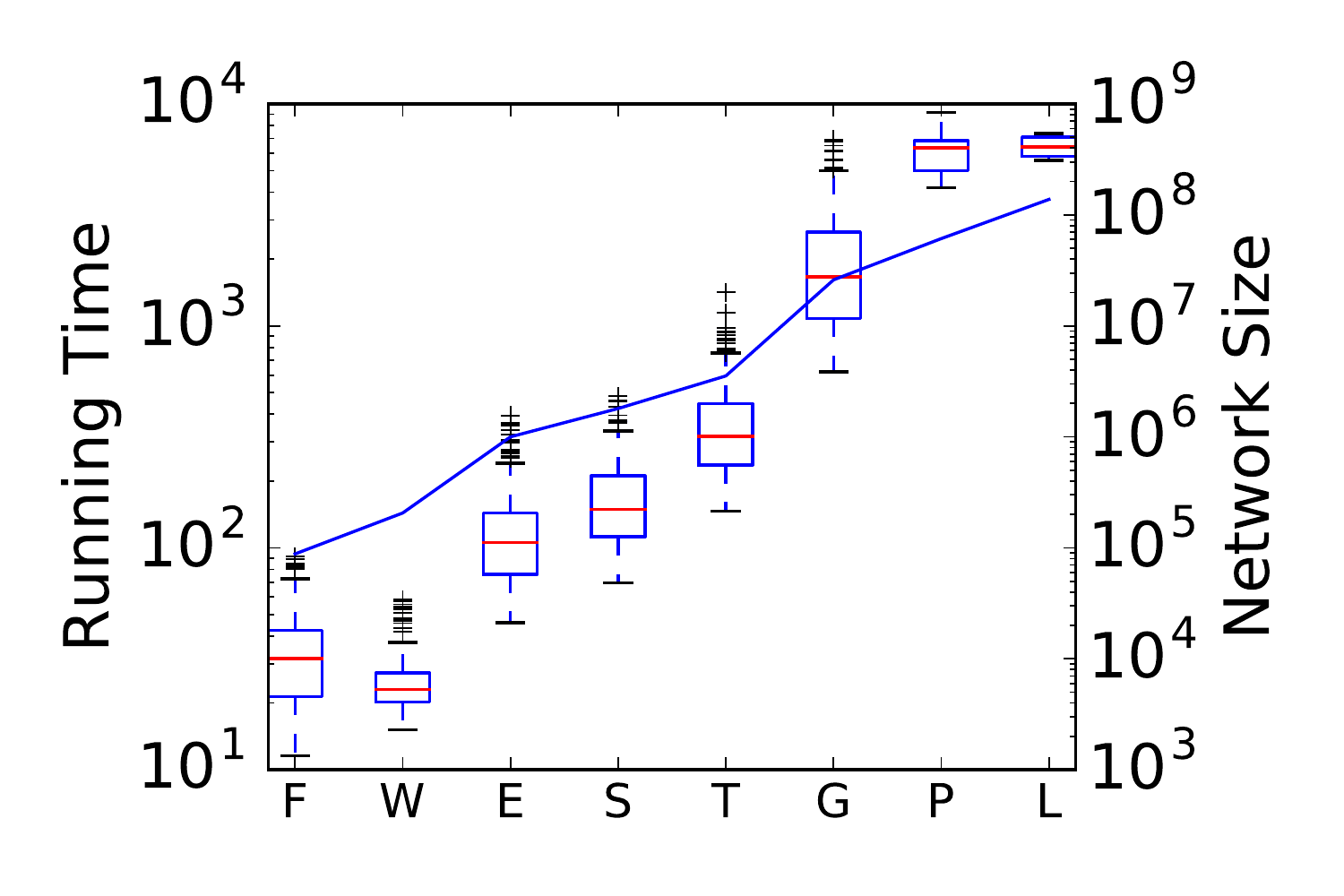}
 	\caption{Running Time of FAST}
 	\label{fig:runningtime}
\end{figure}

\textbf{Scalability of FAST.} During our experiments, we observe that all the scenarios had less than $20$ iterations inside FAST, which means the time complexity of FAST is larger than that of MMinSeed only by a multiplicative constant ($< 40$). Figure \ref{fig:runningtime} demonstrates the scalability of FAST. The trend line denotes the size of different networks (number of edges) and the box plots displays the running time in different networks with various settings. Notice that we use the first letter of each dataset due to space issues.  We can observe a nice feature of FAST that its running time grows linearly in terms of network size. For large networks such as Pokec (61.2M edges) and LiveJournal (138M edges), FAST can finish within three hours.

\subsection{Quality of the Seed Sets.}
Due to space limit, we only display the results for half of the datasets in this section, the results for the remaining are similar. 
Figure \ref{fig:seeddistribution} compares the seed sets obtained by FAST and those obtained by solving the base case of TAP without considering the rate increase. In the pie charts, the shared seeds means the proportion of seed nodes that exists in both FAST seeds and base seeds. The data in each chart is averaged among the result from 180 runs of various settings. Clearly, considering the rate increase explicitly result in a large reduction in number of seeds required.  

\begin{figure*}[!ht]

	\subfloat[Facebook]{
		\includegraphics[width=0.24\linewidth]{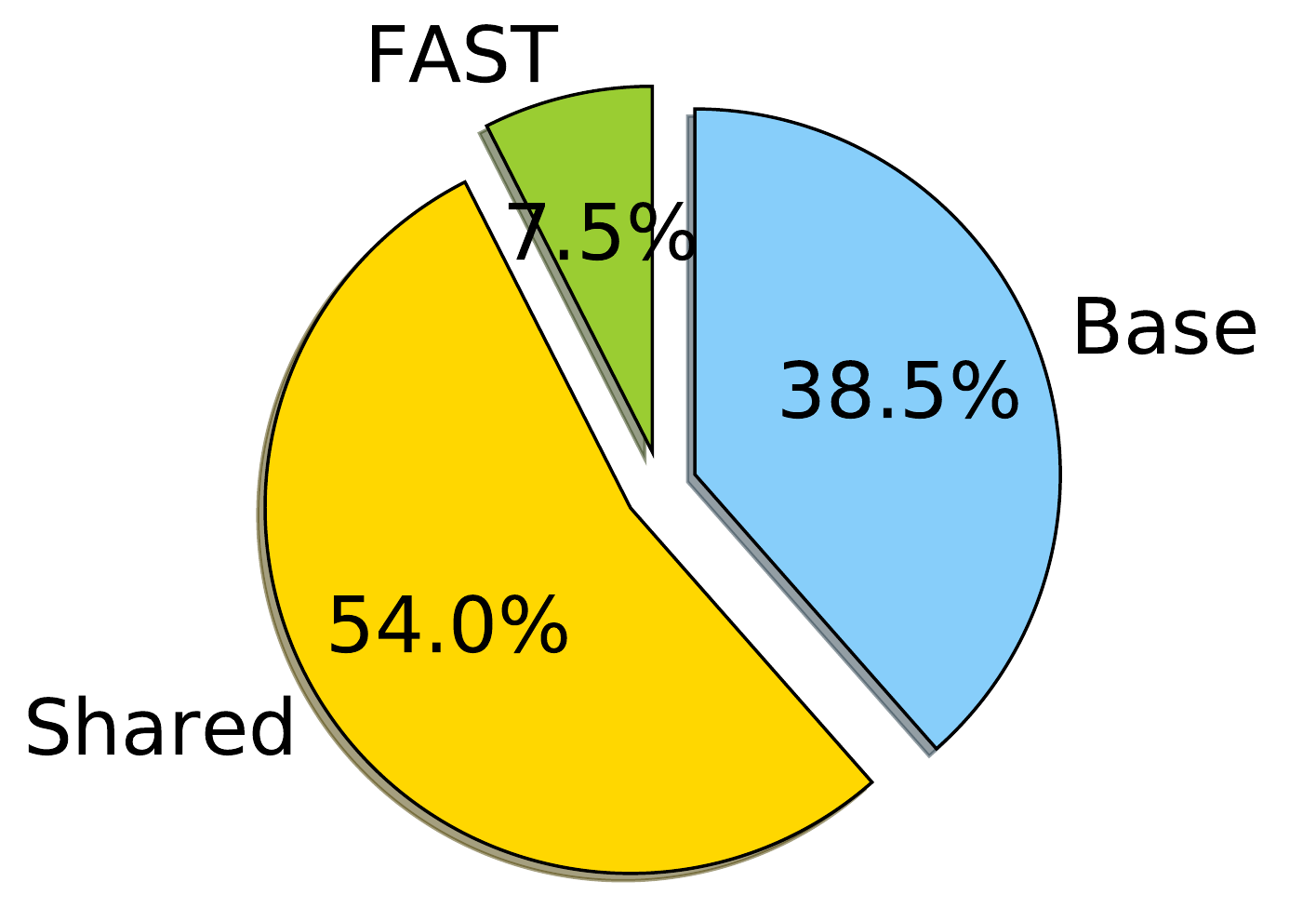}
	}
	\subfloat[Slashdot]{
		\includegraphics[width=0.24\linewidth]{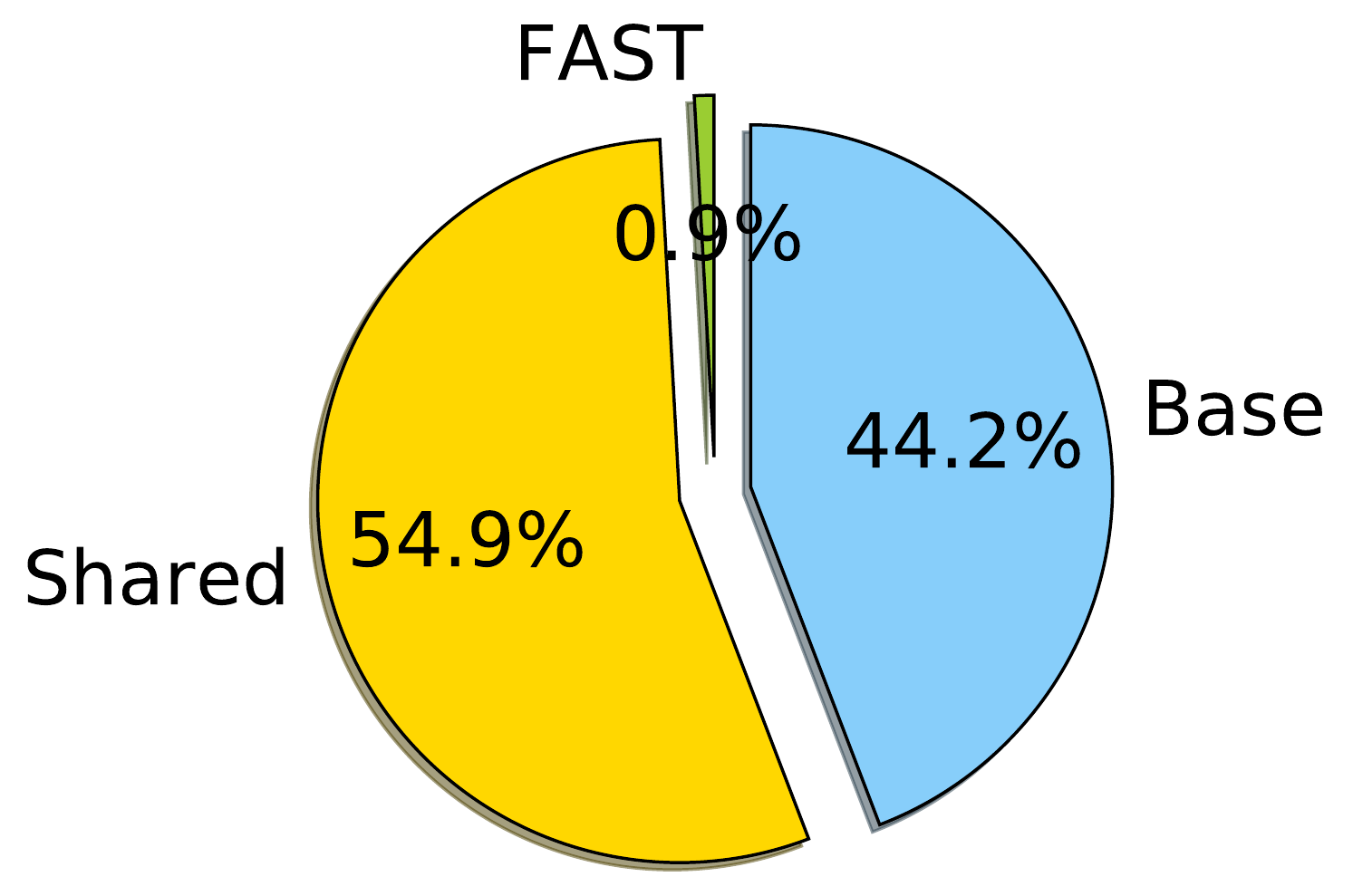}
	}
    \subfloat[Gplus]{
		\includegraphics[width=0.24\linewidth]{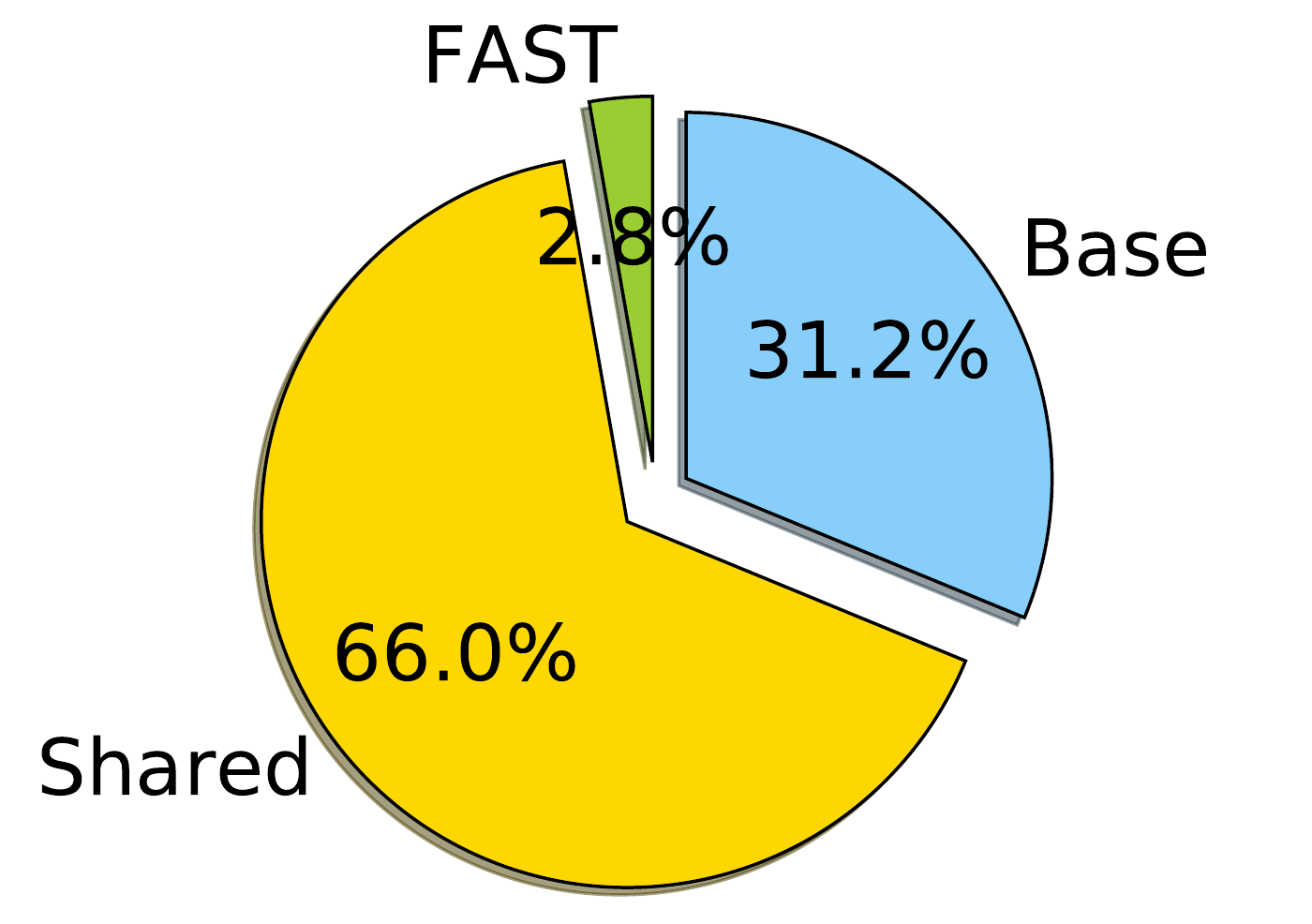}
	}
	\subfloat[Livejournal]{
		\includegraphics[width=0.24\linewidth]{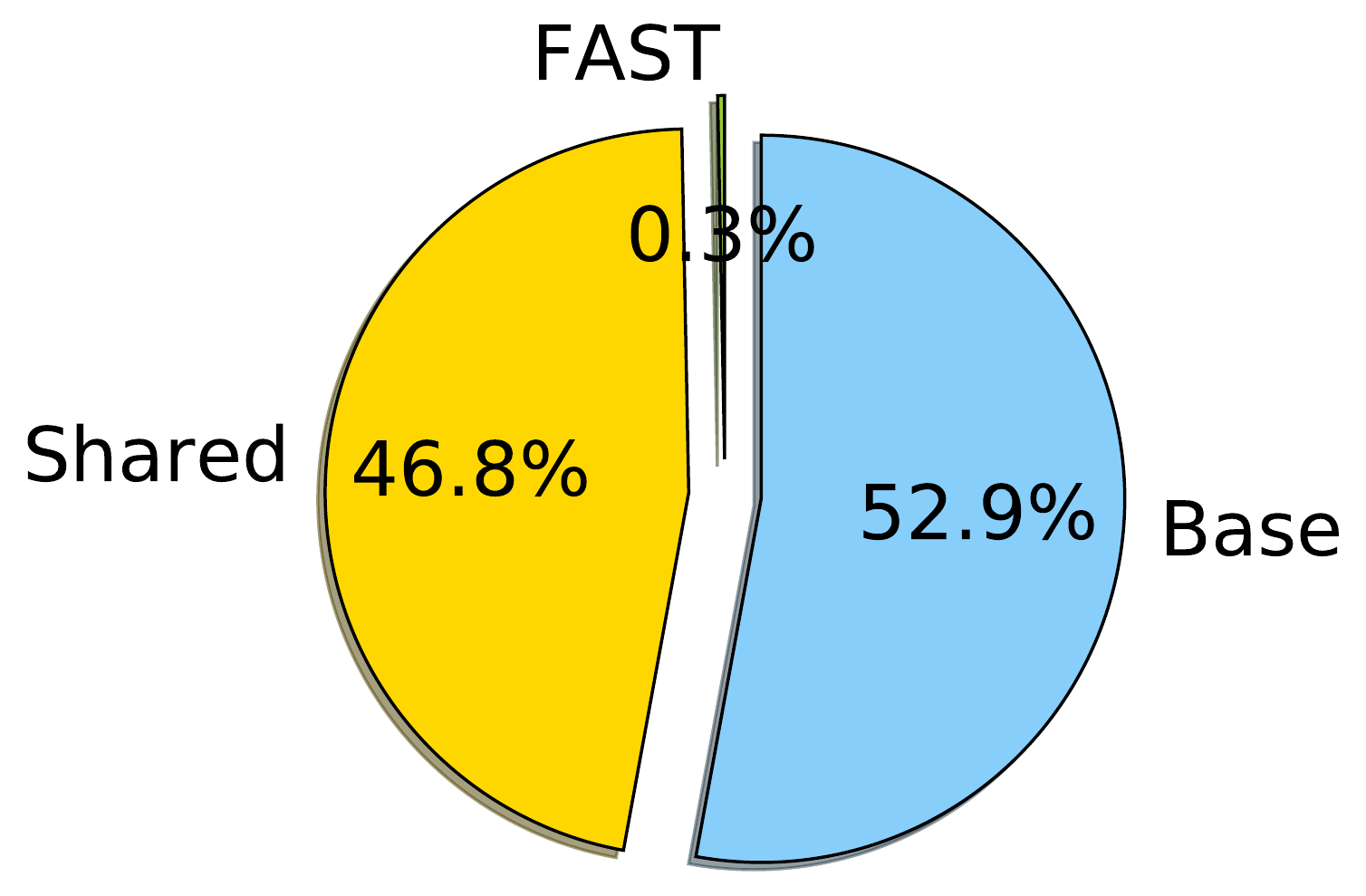}
        }
	\caption{Seed Set Distribution}
	\label{fig:seeddistribution}
\end{figure*}

\begin{figure*}[!ht]

	\subfloat[Facebook]{
		\includegraphics[width=0.24\linewidth]{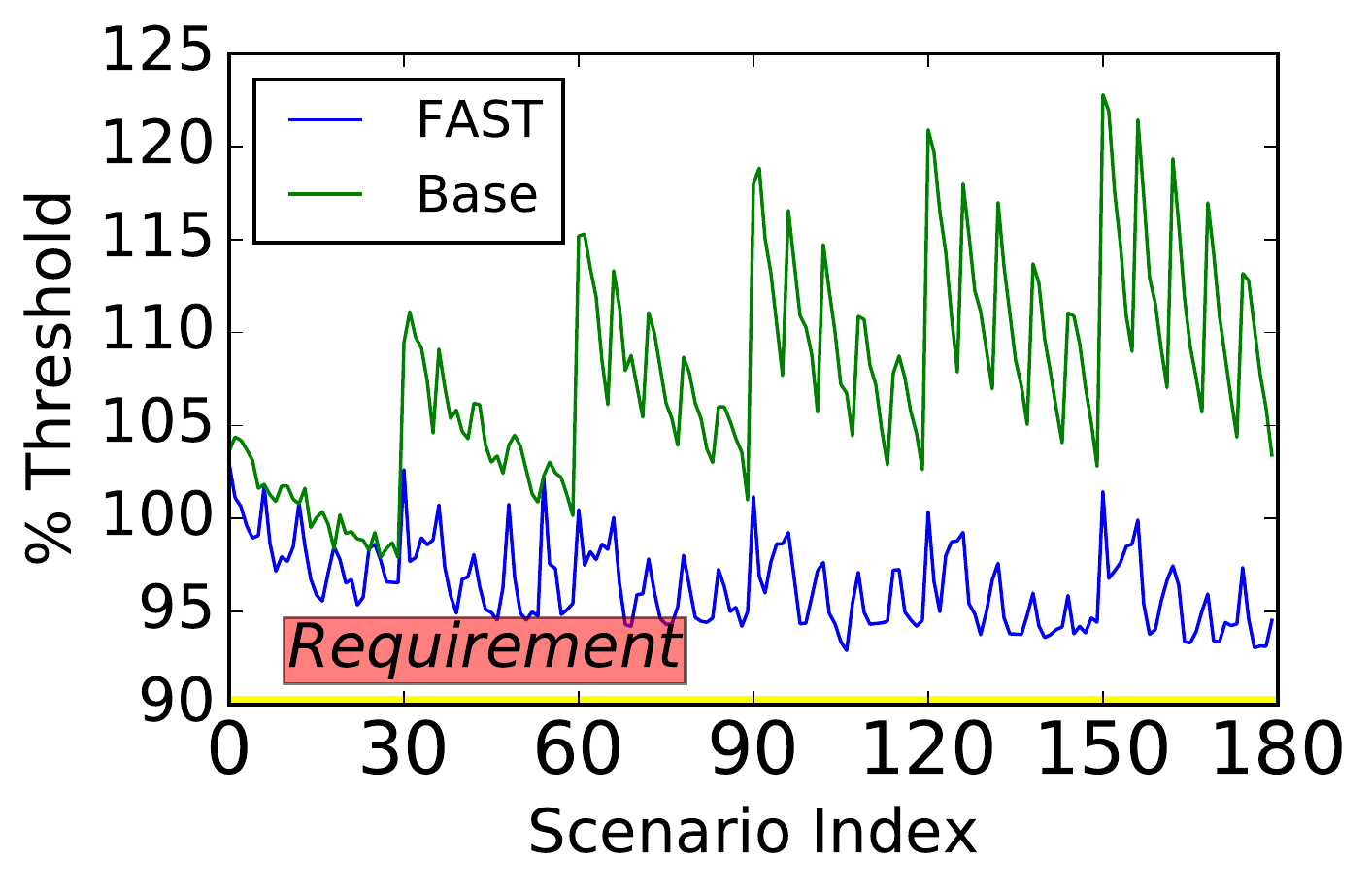}
	}
	\subfloat[Slashdot]{
		\includegraphics[width=0.24\linewidth]{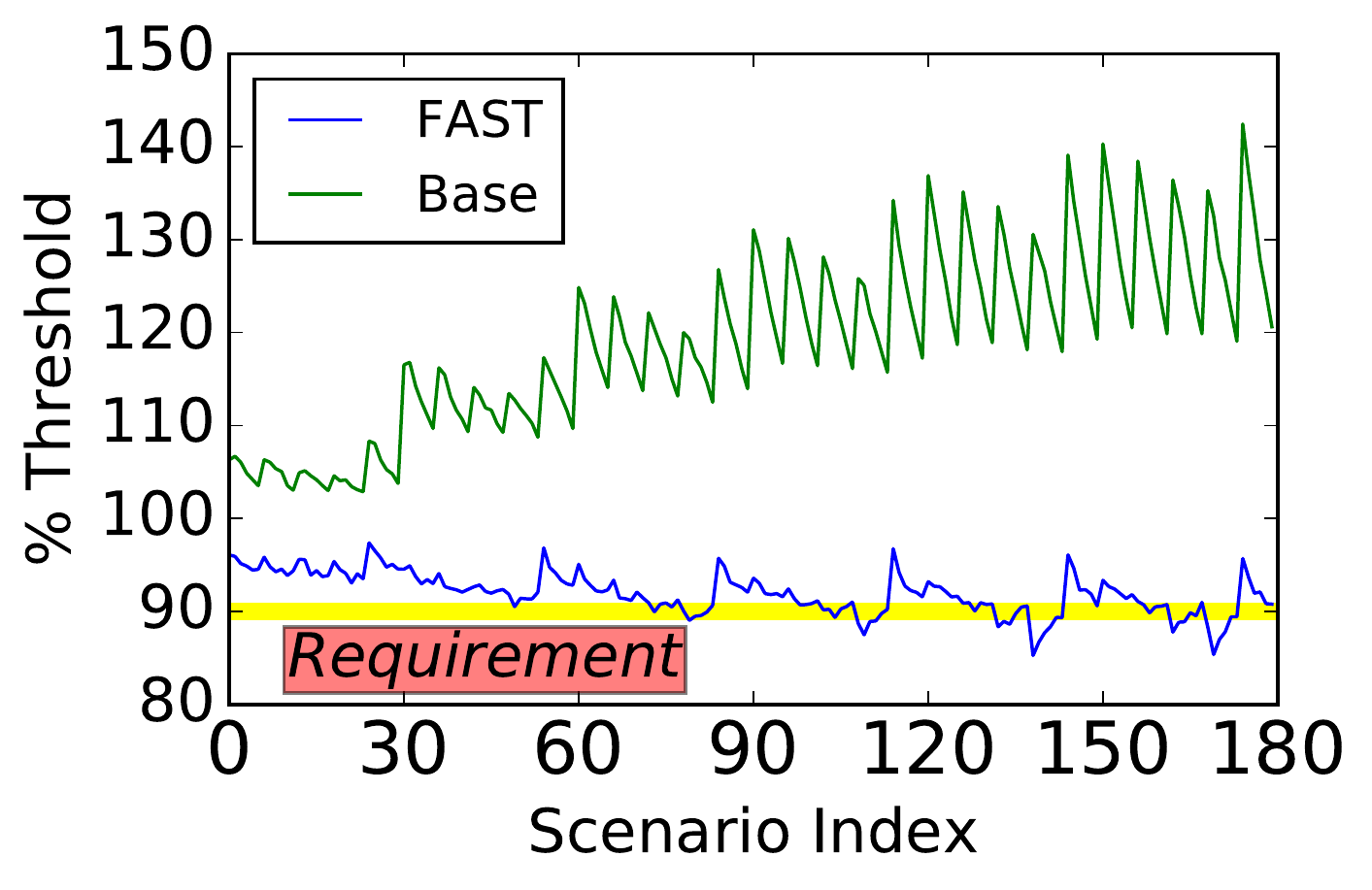}
	}
	\subfloat[Gplus]{
		\includegraphics[width=0.24\linewidth]{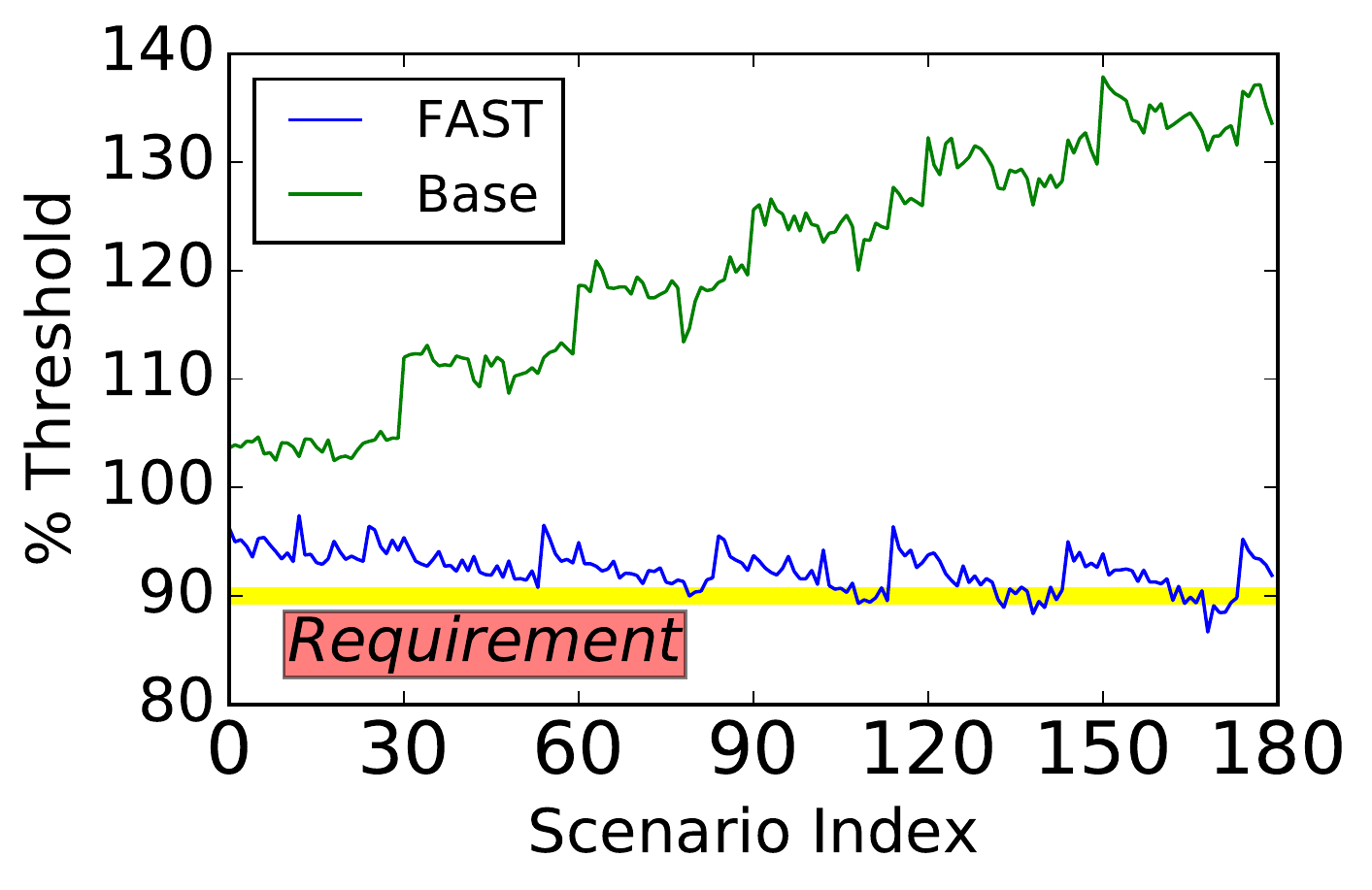}
	}
    \subfloat[LiveJournal]{
		\includegraphics[width=0.24\linewidth]{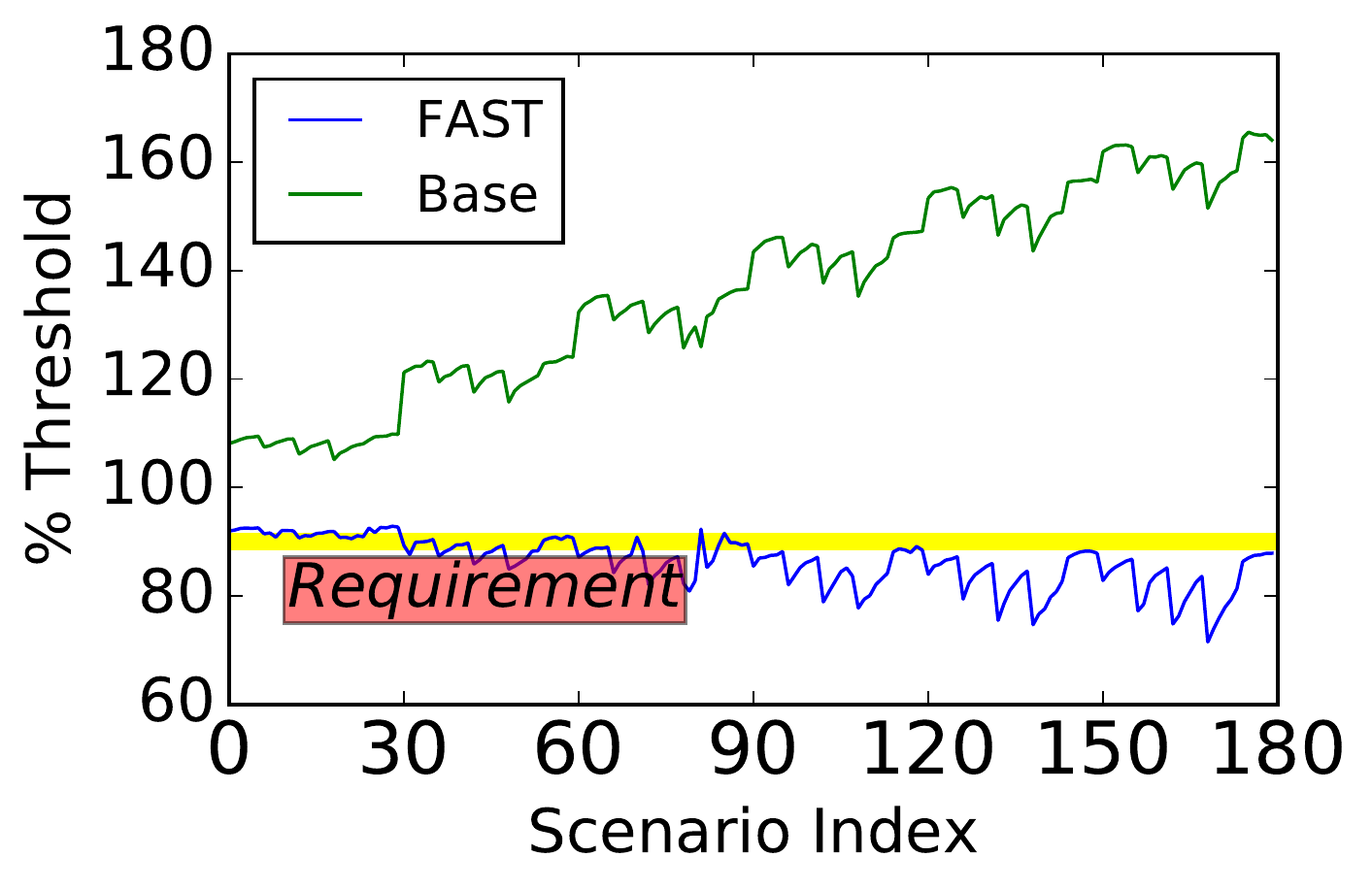}
	}
	\caption{Percentage of Nodes Activated}
	\label{fig:coverage}

\end{figure*}

As the FAST seed set is mostly a subset of the base seed set, we would expect that the FAST seed set has a smaller influence spread. However, we will demonstrate that the FAST seeds can still meet the threshold with extensive simulations. For each setting in each dataset, we simulate $10,000$ random propagation process from both the FAST seeds and the base seeds. The rate increase is applied when the \textit{actual} number of influenced nodes hit the trending thresholds. Figure \ref{fig:coverage} depicts the percentage of nodes activated comparing with the corresponding threshold. As we set $\epsilon$ at $0.1$, the reference line (in yellow) is drawn at $90\%$. A point over the reference line means the average number of activated nodes (over 10,000 simulations) in a certain scenario meets $90\%$ of the threshold. In all the four datasets, the base seed set activated much more nodes than required by the threshold, which suggests that their size can actually be reduced. In most cases, the FAST seed sets can meet the requirement. One exception is in the LiveJournal data set, it is likely that the influence spread in LiveJournal is not concentrated at the expectation.
\subsection{Sensitivity Analysis of Key Parameters}
In this section, we test the impact of three parameters: $r$ (the propagation rate after trending), T-Node (\# nodes required to be trending) and A-Node (overall activation requirement) to the behavior of FAST. The results are displayed in heat maps, a warmer color means an earlier time for rate increase (decided by FAST). The data sets are grouped based on parameter setting (refer to Table \ref{tab:data_sum} for details) into small (Facebook, WikiVote), medium (Gplus, Twitter, Epinions, Slashdot) and large (LiveJournal, Pokec). In each heat map, we vary two parameters and take average over the other. A warmer color in a cell denotes a earlier time for rate increase. 

\begin{figure*}[!ht]
    \subfloat[Small, $r$ and T-Node]{
		\includegraphics[width=0.32\linewidth, height = 35mm]{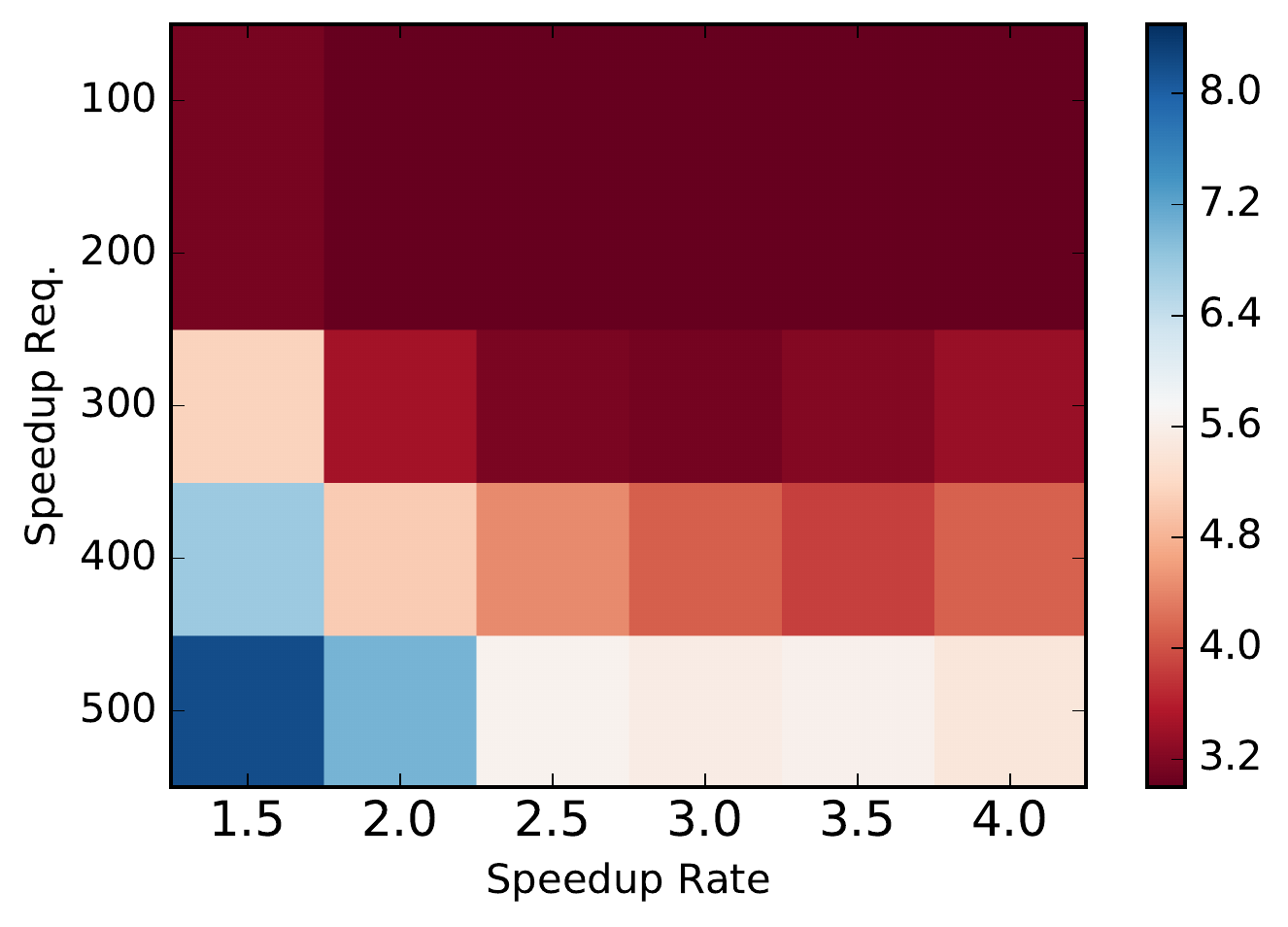}
        \label{fig:ssr1}
	}
	\subfloat[Medium, $r$ and T-Node]{
		\includegraphics[width=0.32\linewidth, height = 35mm]{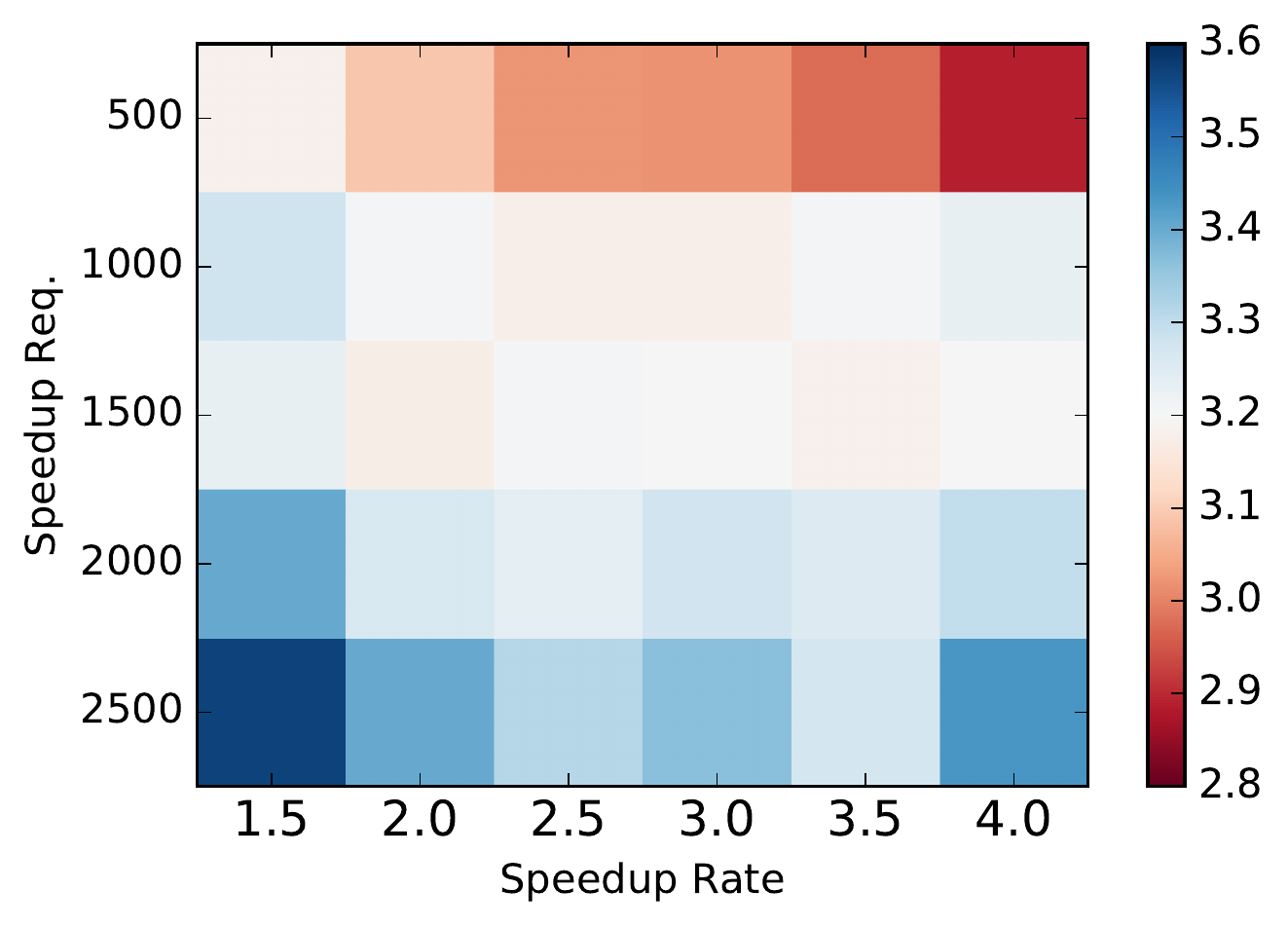}
        \label{fig:ssr2}
	}
    \subfloat[Large, $r$ and T-Node]{
		\includegraphics[width=0.32\linewidth, height = 35mm]{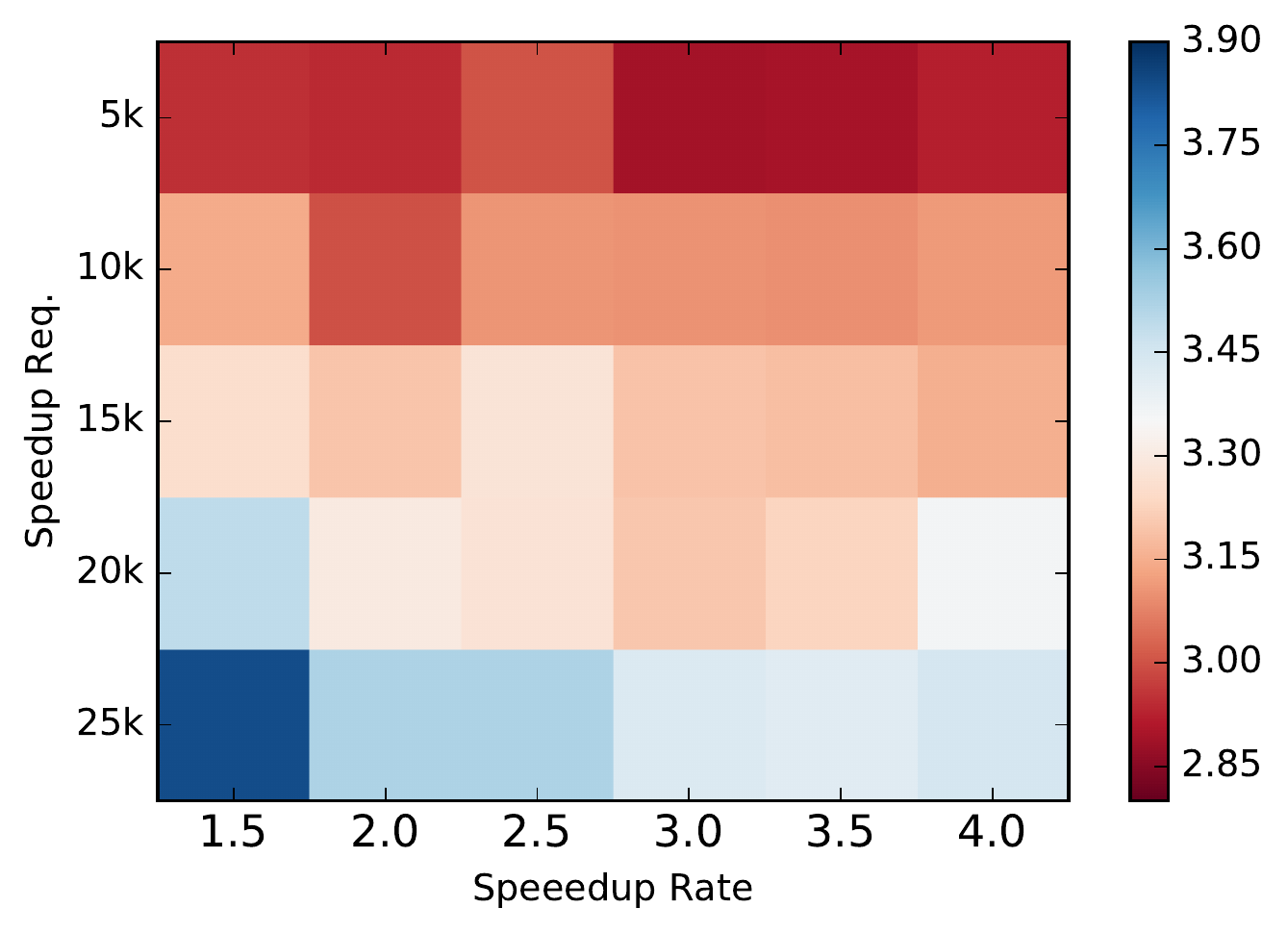}
        \label{fig:ssr3}
	}\\
    \subfloat[Small, T-Node and A-Node]{
		\includegraphics[width=0.32\linewidth, height = 35mm]{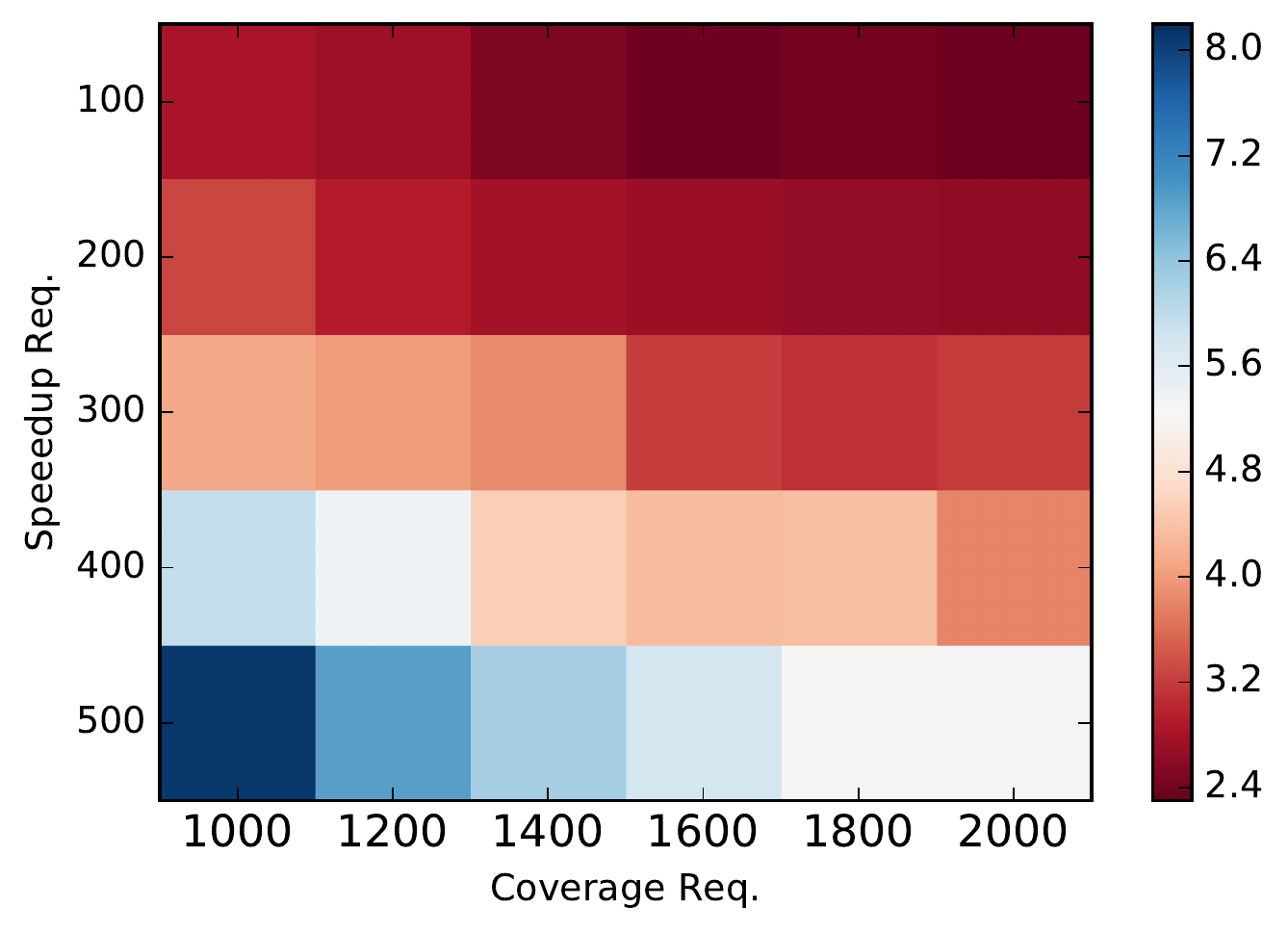}
        \label{fig:srcr1}
	}
    \subfloat[Medium, T-Node and A-Node]{
		\includegraphics[width=0.32\linewidth, height = 35mm]{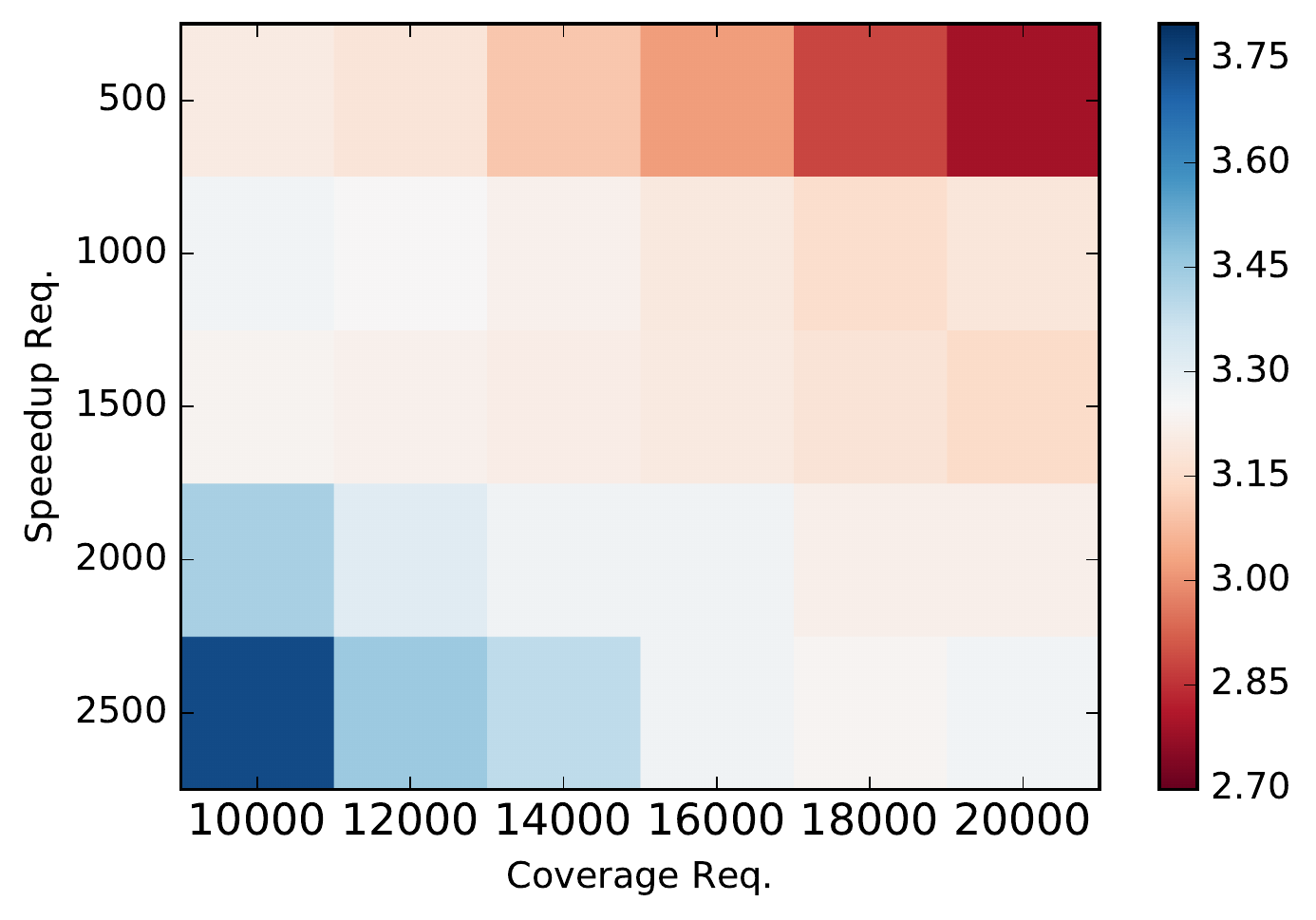}
        \label{fig:srcr2}
	}
    \subfloat[Large, T-Node and A-Node]{
		\includegraphics[width=0.32\linewidth, height = 35mm]{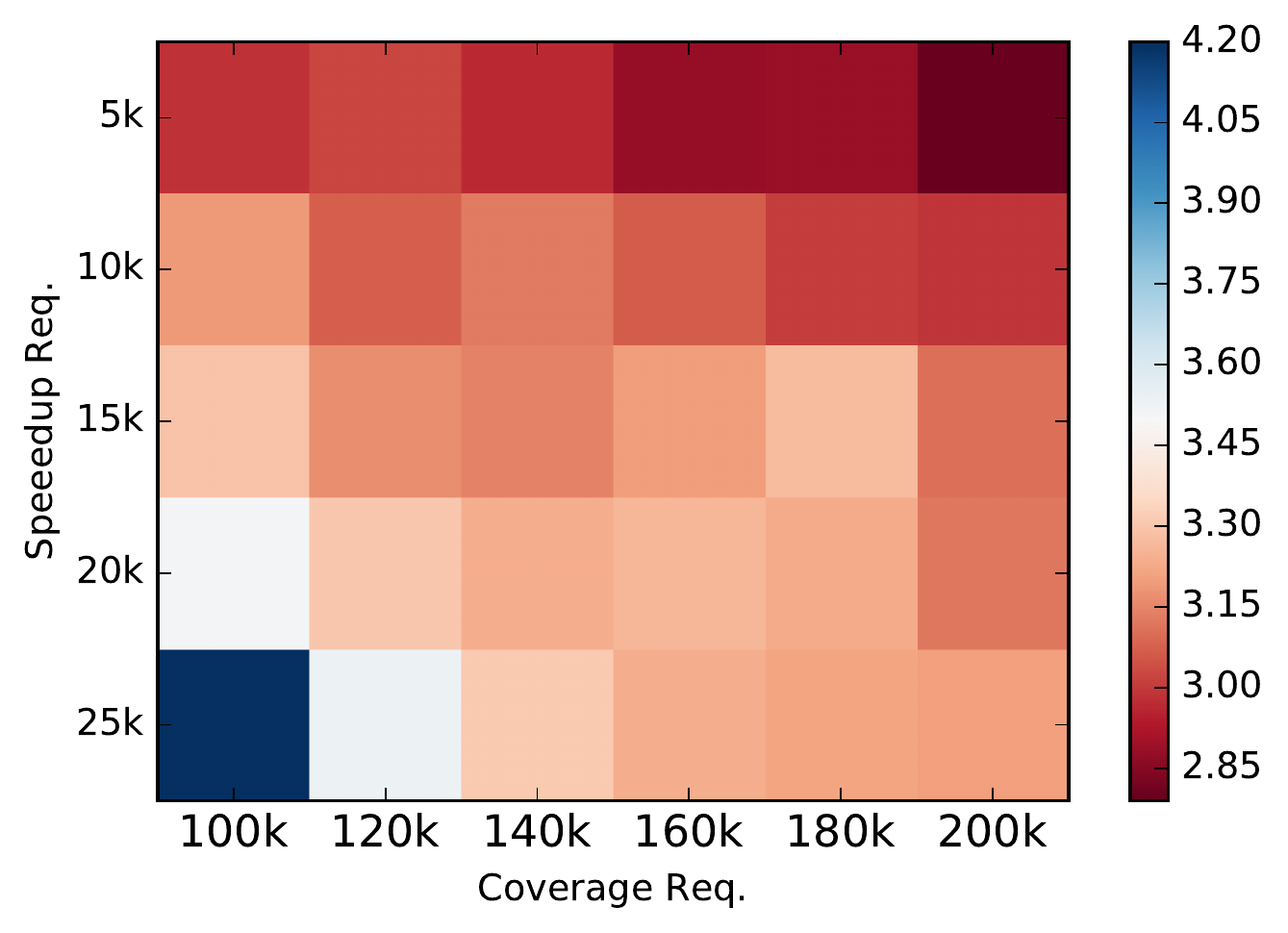}
        \label{fig:srcr3}
	}\\
	\subfloat[Small, $r$ and A-Node]{
		\includegraphics[width=0.32\linewidth, height = 35mm]{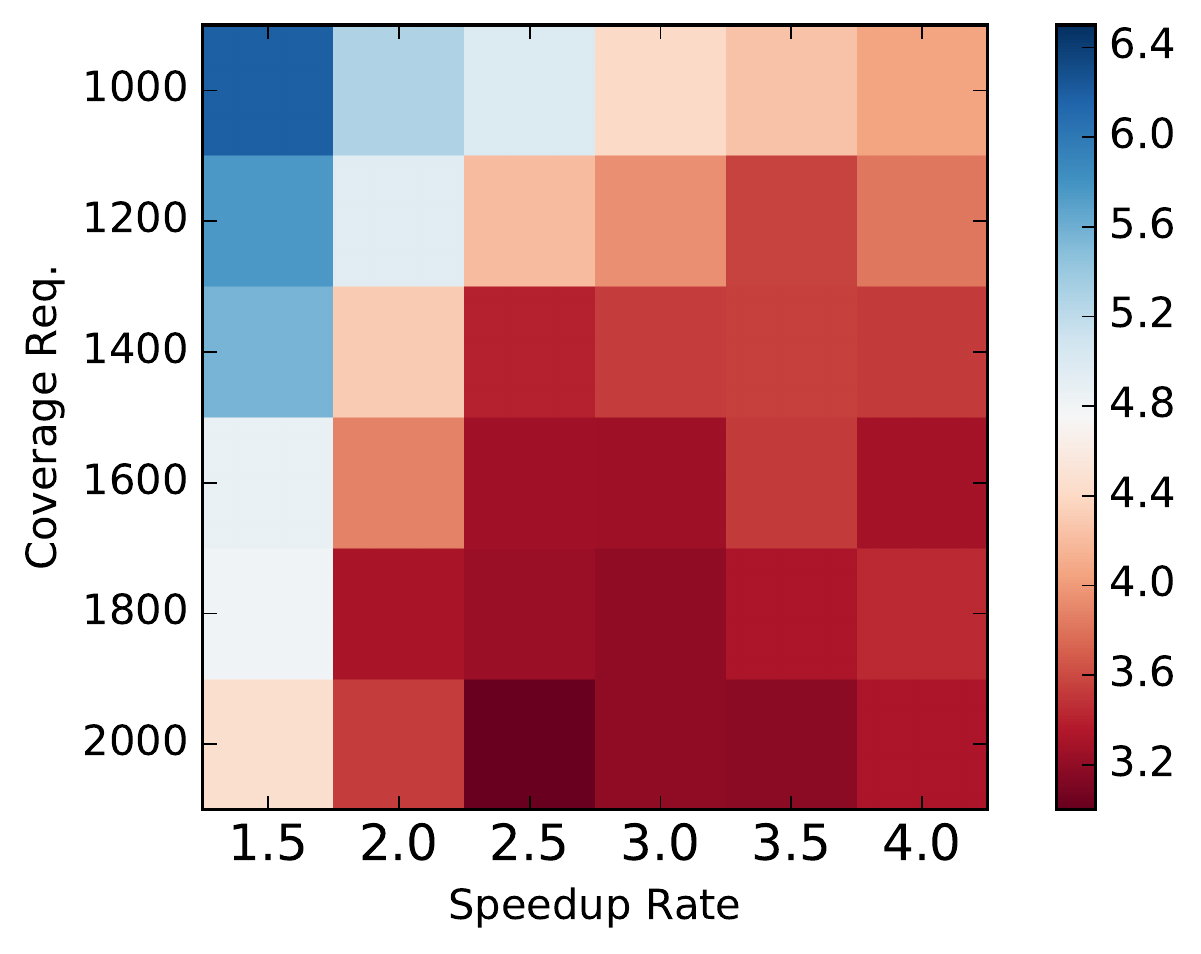}
        \label{fig:scr1}
	}
	\subfloat[Medium, $r$ and A-Node]{
		\includegraphics[width=0.32\linewidth, height = 35mm]{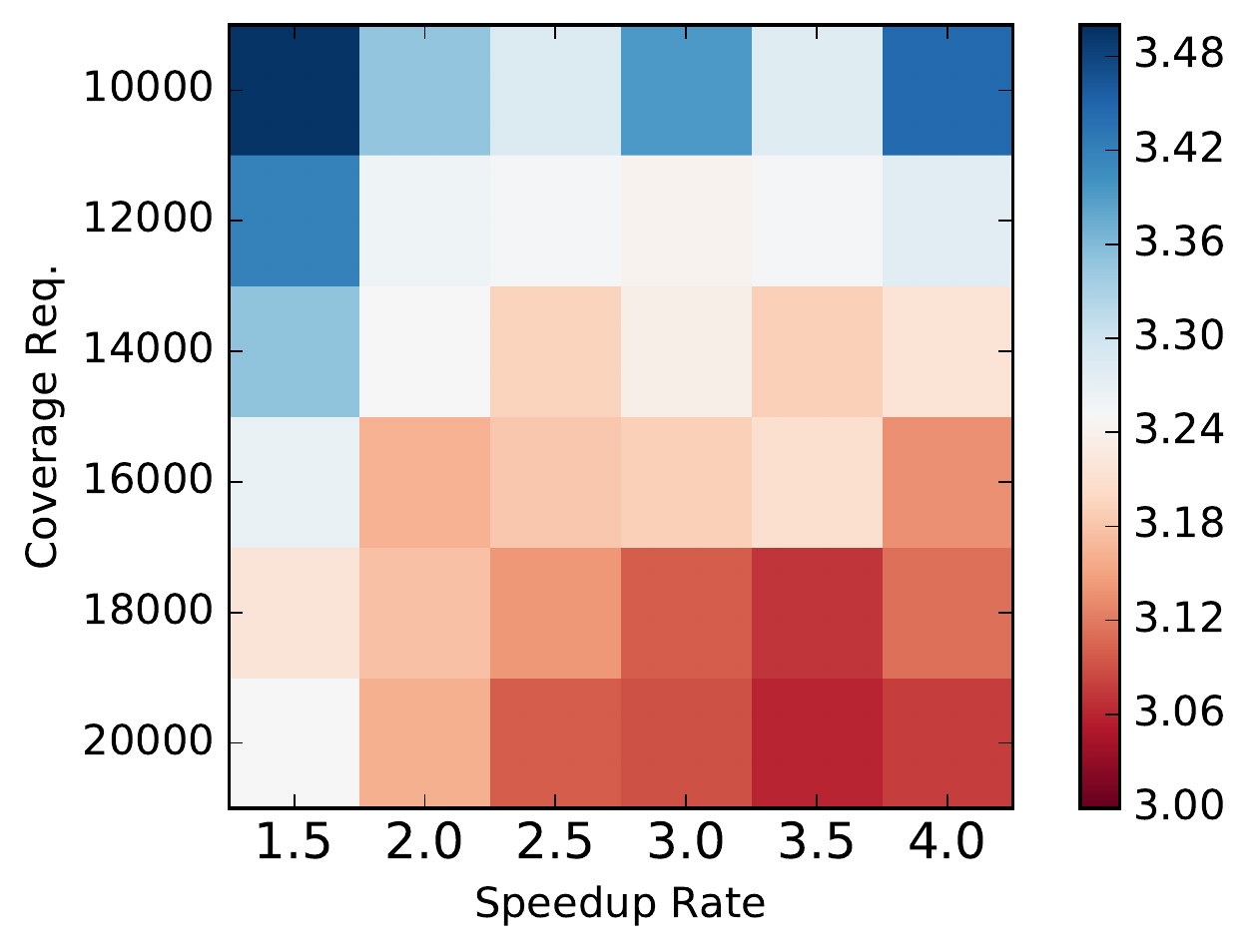}
        \label{fig:scr2}
	}
    \subfloat[Large, $r$ and A-Node]{
		\includegraphics[width=0.32\linewidth, height = 35mm]{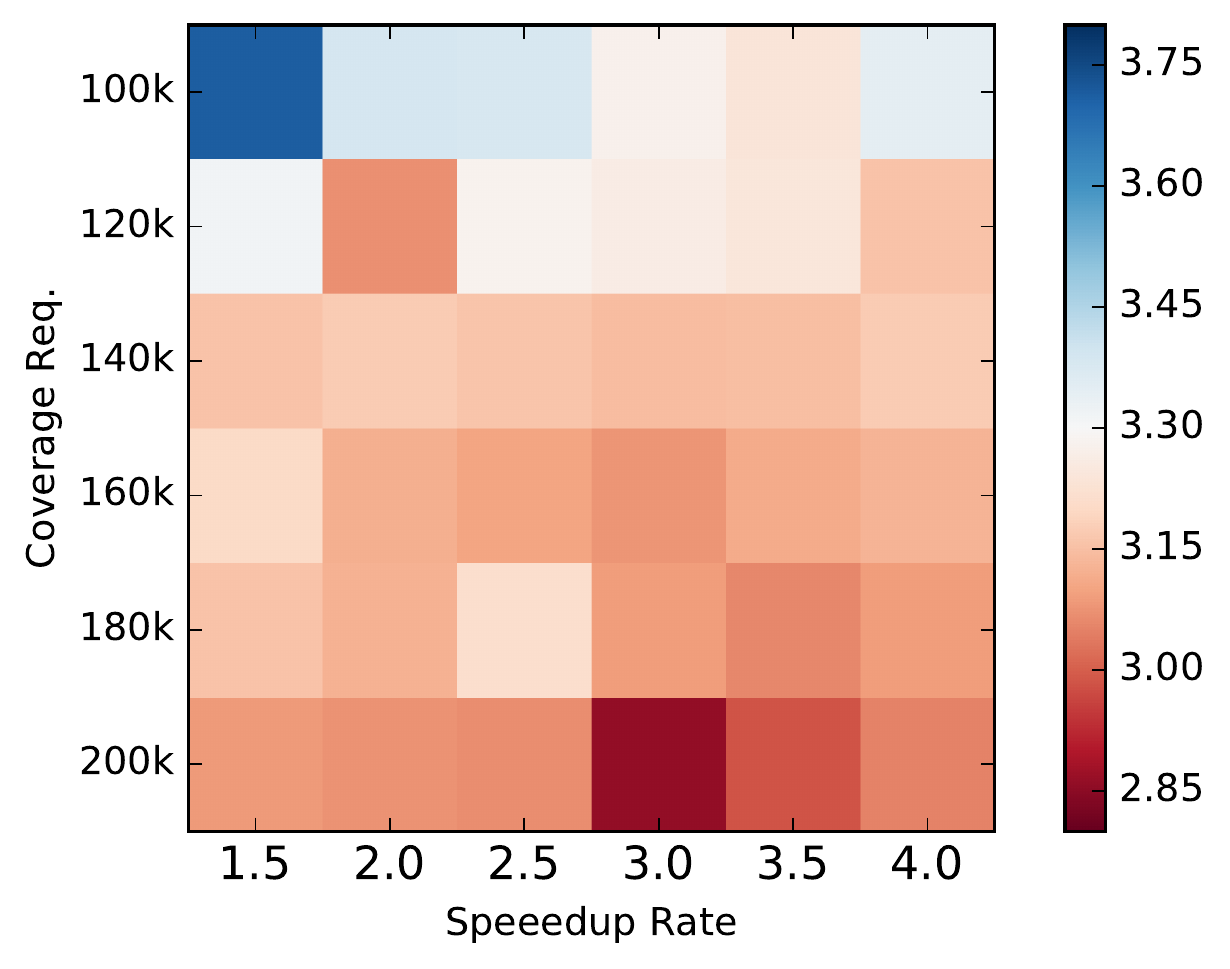}
        \label{fig:scr3}
	}
	\caption{Sensitivity Analysis Results}
	\label{fig:sensitivity}
\end{figure*}
In Fig.~\ref{fig:ssr1}-\ref{fig:ssr3}, we vary $r$ and T-Nodes. One clear trend is that when more T-Nodes are required, FAST tend to choose a later time for starting the trend and the speed up. It is reasonable as it can be costly to activate T-Nodes early when the number is high. As we observed in some scenarios, FAST may choose not to have the speed up at all when T-Node is too large and $r$ is small. What seems counter-intuitive in those Figures is that 
FAST may not choose to have the speed up earlier when $r$ increases. This phenomena can possibly explained from the perspective of cost. With a larger $r$, FAST can influence a set of nodes in less time, even with the same seed set. Thus, FAST is not that ``hurry" of starting the speed up, as it is able to reach the threshold when triggering the speed up later with less cost. 

In Figs.~\ref{fig:srcr1}-\ref{fig:scr3}, we vary A-Nodes. In most of the cases, FAST tend to have the speed up earlier when facing a larger threshold. The opposite is shown in some rare cases, mostly because $r$ is too low or T-Node is too high, which makes FAST think starting the speed-up early is not beneficial. 

\section{Conclusion and Future Works}\label{sc:conclusion}
In this paper, we proposed a novel dynamic influence propagation model, which can more accurately characterize the information diffusion in social networks compared with the existing propagation models. The model is supported by analysis of the crawled retweet data, that most topics will propagate faster after being trending. To study the impact of DIP in OSNs, we propose a new TAP-DIP problem by substituting the static propagation model in TAP with DIP. Although TAP-DIP is even harder than TAP, we designed the FAST algorithm that can solve it with approximation ratio similar to the best for TAP. In experiments, we demonstrated that FAST can generate high quality seed sets and is scalable. Also, we confirmed that DIP has a high impact in the result.

The considered TAP-DIP problem allows a single rate increase, yet the propagation rate of some topics may change multiple times, or even continuously. Also, TAP-DIP requires activating one threshold, while some variations of TAP may need to activate different thresholds for different target groups, especially for multiplex OSNs or OSNs with clear community structure \cite{nguyen2011overlapping}. Notice that if we can fix all the times for rate increase, the MMinSeed algorithm is still applicable, as it is general enough to handle multiple thresholds at the same time. However, the performance of the FAST algorithm may be degraded, since we need to locate not one, but multiple times for rate increase, and multi-variate Lipschitz Optimization can be time-consuming. To tackle this problem, we may allow heuristic algorithms that focus on one rate increase time per iteration and assume all other times are fixed. We leave the design of such heuristic algorithms, or even a completely new algorithm to solve TAP-DIP, in future work. 
\section*{Acknowledgement}
This work was supported by NSF CCF 1422116.
 \bibliographystyle{unsrt}

\appendix

\section{The Proofs}\label{sc:appendixproofs}
\begin{proof}[Proof of Lemma~\ref{lemma:ddipsinglebound}]\let\qed\relax
First, it is clear that $H'(t^*)>H'(t^{*'})$ and $H(t^*)<H'(t^*)$ by definition of optimal values of $H(.)$ and $H(.)'$. 

We also claim that $H(t^{*})\geq \max(S^*_s(t^{*}),S^*_a(t^{*}))$ as otherwise, at least one of $S^*_s(t^{*}), S^*_a(t^{*})$ is not minimum. Therefore, 
\begin{align*}
\frac{H'(t^{*'})}{H(t^*)}&\leq \frac{H'(t^{*})}{H(t^*)}\leq \frac{2 \max(S^*_s(t^*),S^*_a(t^*))}{H(t^*)}\\
&\leq \frac{2 \max(S^*_s(t^*),S^*_a(t^*))}{\max(S^*_s(t^*),S^*_a(t^*))} = 2\qquad\QED
\end{align*}
\end{proof}

\begin{proof}[Proof of Lemma \ref{lemma:locallipschitzconstant}]\let\qed\relax
Denote $t'<t''$ as two arbitrary points in $[t_1,t_2]$. As $S^*_s(t), S^*_a(t)$ are monotone decreasing/increasing respectively, we have
\begin{align*}
0&\leq S^*_s(t')-S^*_s(t'')\leq S^*_s(t_1)-S^*_s(t_2)\\
0&\leq S^*_a(t'')-S^*_a(t')\leq S^*_a(t_2)-S^*_a(t_1)
\end{align*}

Therefore, 
\begin{align*}
H'(t')-H'(t'') &= (S^*_s(t')-S^*_s(t''))-(S^*_a(t'')-S^*_a(t'))\\
			&\leq S^*_s(t')-S^*_s(t'') \leq S^*_s(t_1)-S^*_s(t_2)
\end{align*}

Similarly,
\begin{align*}
H'(t')-H'(t'') &\geq -(S^*_a(t'')-S^*_a(t'))\geq -(S^*_a(t_2)-S^*_a(t_1))
\end{align*}

Thus, 
\begin{align*}
|H'(t')-H'(t'')|&\leq \max \{S^*_s(t_1)-S^*_s(t_2),S^*_a(t_2)-S^*_a(t_1)\}\\
			&\leq l_{t_2}|t'-t''| \text{\quad (Since  $|t-t''|\geq \Delta$)}\quad\QED
\end{align*}
\end{proof}

\begin{proof}[Proof of Theorem \ref{theorem:lipschitzratio}]\let\qed\relax
By the time Alg. \ref{alg:lipschitz} stops, the global minimum $H'(t^*)$ of $H'(.)$ is lower bounded by 
\begin{align*}
R_i=\frac{H'(t_{i-1})+H'(t_i)}{2}-\frac{l_i(t_i-t_{i-1})}{2}
\end{align*}
as $R_i$ is the minimum lower bound of all intervals. Since $H'(\bar{t})$ is the minimum of all calculated function values, we have:
\begin{align*}
& H(\bar{t})-2H(t^*)  \quad \text{($H(.)$ is lower bounded by $H'(.)$)}\\
& \leq H'(\bar{t})-H'(t^{*'})\leq \min\{H'(t_{i-1}),H'(t_i)\}-R_i\\
				&= |\frac{H'(t_{i-1})-H'(t_i)}{2}|+\frac{l_i(t_i-t_{i-1})}{2}\\
                &\leq \frac{l_i(t_i-t_{i-1})}{2}+\frac{l_i(t_i-t_{i-1})}{2} \quad (\text{Lipschitz Continuity})\\
                &\leq 1 \quad(|t_i-t_{i-1}|\leq \frac{1}{l_i})\quad\QED
\end{align*}
\end{proof}

\begin{proof}[Proof of Theorem \ref{theorem:greedybicriteriaratio}]
Denote the optimal solution to MTAP as $S^*$. We consider the greedily selected collection $S_j^g$ with smallest $j$ that satisfies
\begin{align*}
f(S_j^g)\geq f(S^*) -c
\end{align*}
where $c$ is a constant. So, we want to solve the equation $f(S_j^g)= f(S^*) -c$. By \eqref{eqn:greedyguarantee}, we have
\begin{align*}
(1-(1-\frac{1}{|S^*|})^j-\epsilon)f(S^*)&\leq f(S^{*})-c\\
-(1-\frac{1}{|S^*|})^j&\leq -\frac{c-\epsilon f(S^*)}{f(S^*)}\\
-j \log (1-\frac{1}{|S^*|}) &\leq \log f(S^*)-\log (c-\epsilon f(S^*))\\
\frac{j}{|S^*|}&<\log f(S^*)-\log (c-\epsilon f(S^*))
\end{align*}
The last inequality follows as $\log (1-\frac{1}{|S^*|})<-\frac{1}{|S^*|}$ and $\frac{1}{|S^*|}\in (-1,0)$. By selecting $c=\epsilon f(S^*)$, we have $\frac{j}{|S^*|}<\log f(S^*)$ and  $f(S_j^g) >= (1-\epsilon) f(S^*)$.
\end{proof}

\begin{proof}[Proof of Lemma \ref{lemma:numberofsamples}]
We write $f(S) = \sum_{l=1}^L f^l(S)$ where $f^l(S)=\min\{\eta_l|V^l|, \mathbb{I}_T^l(S)\}$. Denote $\mathcal{R}^l$ as the collection of samples for estimating $\mathbb{I}_T^l(S)$ and 
\begin{align*}
\hat{\mathbb{I}}_T^l(S)=\frac{\sum_{R_i\in |\mathcal{R}^l|}x_{R_i}}{|\mathcal{R}^l|}|V^l|
\end{align*}
where $x_{R_i}=1$ when $R_i\cap S\neq \emptyset$ and 0 otherwise. Define 
\begin{align*}
\hat{f}(S)=\sum_{l=1}^L\hat{f}^l(S)=\min \{\eta_l|V^l|,\hat{\mathbb{I}}_T^l(S)\}
\end{align*} Denote the optimal size $k$ solution that maximizes $f(S)$ as $S_k^*$ and the greedy size $j$ solution as $S_j^g$.

Now, we bound the number of samples required to approximate $f(S_k^*)$ using $\hat{f}(S_j^g)$ with $(1-(1-1/k)^j-\epsilon)$ ratio and arrive at \eqref{eqn:greedyguarantee}.
 
We will need the following two lemmas in our proof.
\begin{lemma}\label{lemma:firstthreshold} \cite{Tang15}
	Let $\delta_1\in(0,1),\epsilon_1>0$, and 
	\begin{align}
		Q_1^l = \frac{2|V^l|\log(\frac{1}{\delta_1})}{\mathbb{I}^l_T(S_k^*)\epsilon_1^2} \label{eqn:firstthreshold}
	\end{align}
	If the total number of samples is at least $Q_1^l$, then $\hat{\mathbb{I}}^l_T(S_k^*)\geq (1-\epsilon_1)\mathbb{I}^l_T(S_k^*)$ holds with at least $1-\delta_1$ probability, for the optimal seed set $S^*$.  
\end{lemma}
\begin{lemma}\label{lemma:secondthreshold} \cite{Tang15}
	Let $\delta_2\in(0,1), \epsilon>\epsilon_1$ and 
	\begin{align}
		Q_2=\frac{2(1-(1-1/k)^j)|\mathcal{V}|(\log\binom{|\mathcal{V}|}{j}+\log\frac{1}{\delta_2}))}{f(S_k^*)(\epsilon-(1-(1-1/k)^j)\epsilon_1)^2}\label{eqn:secondthreshold}
	\end{align}
	For the set $S^g_j$, if 
	\begin{align}
		\hat{f}(S^g_j)\geq (1-(1-1/k)^j)(1-\epsilon_1) \hat{f}(S_k^*)\label{eqn:greedycondition}
	\end{align}
	and the total number of samples is at least $Q_2$, then $f(S^g_j)\geq (1-(1-1/k)^j-\epsilon) f(S_k^*)$ holds with at least $1-\delta_2$ probability.
\end{lemma}
The original Lemma \ref{lemma:secondthreshold} is only for the case when $j=k$ and for function $\mathbb{I}_T(S)$ instead of $f(S)$, yet it can be easily extended to this more generalized version. Also, since we consider situations with $j>k$ and the first derivative of $Q_2(k)$ is positive, we can increase $k$ such that $(1-1/k)^j = 1/e$ and achieve a upper bound $\bar{Q}_2$ on $Q_2$ with single variable $j$.
	\begin{align}
		\bar{Q}_2=\frac{2(1-1/e)|\mathcal{V}|(\log\binom{|\mathcal{V}|}{j}+\log\frac{1}{\delta_2}))}{f(S_k^*)(\epsilon-(1-1/e)\epsilon_1)^2}\label{eqn:secondthresholdupperbound}
	\end{align}
We now claim that $\hat{f}^l(S_k^*) \geq (1-\epsilon_1)f^l(S_k^*)$ holds with probability $(1-\delta_1)$ with at least $Q_1^l$ samples. 

To prove the claim, we consider four situations w.r.t. $\mathbb{I}^l_T(S),\hat{\mathbb{I}}^l_T(S)$ and $\eta^l$ for each $l$. 

\textbf{Situation 1:} $\hat{\mathbb{I}}_T^l(S_k^*)\leq \eta^l|V^l|, \mathbb{I}_T^l(S_k^*)\leq \eta^l|V^l|$. In this situation, $\hat{f}^l(S_k^*)=\hat{\mathbb{I}}_T^l(S_k^*)$ and $f^l(S_k^*)=\mathbb{I}_T^l(S_k^*)$ and the claim holds naturally due to Lemma \ref{lemma:firstthreshold}.

\textbf{Situation 2:} $\hat{\mathbb{I}}_T^l(S_k^*)\leq \eta^l|V^l|, \mathbb{I}_T^l(S_k^*)> \eta^l|V^l|$. In this situation, $\hat{f}^l(S_k^*)=\hat{\mathbb{I}}_T^l(S_k^*)$ and $f^l(S_k^*)=\eta^l |V^l|<\mathbb{I}_T^l(S_k^*)$. By Lemma \ref{lemma:firstthreshold},
\begin{align*}
\hat{f}^l(S_k^*)&=\hat{\mathbb{I}}_T^l(S_k^*)\geq (1-\epsilon_1)\mathbb{I}_T^l(S_k^*)\\
&\geq (1-\epsilon)\eta^l|V^l| =(1-\epsilon_1)f^l(S_k^*)
\end{align*}

\textbf{Situation 3:} $\hat{\mathbb{I}}_T^l(S_k^*)> \eta^l|V^l|, \mathbb{I}_T^l(S_k^*)\leq \eta^l|V^l|$. In this situation, $\hat{f}^l(S_k^*)=\eta^l|V^l|$ and $f^l(S_k^*)=\mathbb{I}_T^l(S_k^*)\leq \eta^l|V^l|$. So, $\hat{f}^l(S_k^*)>f^l(S_k^*)>(1-\epsilon_1)f^l(S_k^*)$.

\textbf{Situation 4:} $\hat{\mathbb{I}}_T^l(S_k^*)> \eta^l|V^l|, \mathbb{I}_T^l(S_k^*)> \eta^l|V^l|$. In this situation, $\hat{f}^l(S_k^*)=\eta^l|V^l|=f^l(S_k^*)$, so $\hat{f}^l(S_k^*)=f^l(S_k^*)>(1-\epsilon_1)f^l(S_k^*)$.

Combining the results from all four possible situations, we have proved the claim. 

With both Lemma \ref{lemma:secondthreshold} and the claim, the minimum number of samples to guarantee \eqref{eqn:greedyguarantee} is $\max\{\sum_{l=1}^L Q_1^l,\bar{Q}_2\}$. We write 
\begin{align*}
&\bar{Q}_2\leq\sum_{l=1}^L\bar{Q}_2^l=\frac{2(1-1/e)|V^l|(\log\binom{|\mathcal{V}|}{j}+\log\frac{1}{\delta_2}))}{\mathbb{I}^l_T(S_k^*)(\epsilon-(1-1/e)\epsilon_1)^2}\\
&\max\{\sum_{l=1}^L Q_1^l,\bar{Q}_2\}\leq \sum_{l=1}^L\max\{Q_1^l,\bar{Q}_2^l\}=Q_{min}
\end{align*}
By \cite{Tang15}, set
$\delta_1=\delta_2=\frac{\delta}{3L}, \epsilon_1=\epsilon\frac{\sigma}{(1-1/e)\sigma + \tau}$ where
$\sigma =\sqrt{\ln(\frac{3L}{\delta})}, \tau = \sqrt{(1-\frac{1}{e})(\ln\binom{|\mathcal{V}|}{j}+\ln\frac{3L}{\delta})}
$
bring the number of samples to 
\begin{align*}
Q =\sum_{l=1}^L\frac{2|V^l|\phi^2}{\mathbb{I}_T^l(S_k^*)},  \quad\phi = \frac{(1-1/e)\sigma+\tau}{\epsilon}
\end{align*}
which is larger than $Q_{min}$ by at most a small constant. The probability of ensuring \eqref{eqn:greedyguarantee} is at least $1-\frac{L+1}{3L}\delta$ by union bound.
\end{proof}

\begin{proof}[Proof of Lemma \ref{lemma:numberofsampleguarantee}]\let\qed\relax
We introduce Lemma \ref{lemma:martingale} for this proof.
\begin{lemma}\label{lemma:martingale}\cite{Tang15}
For any $\epsilon>0$, 
\begin{align}
Pr[\sum_{i=1}^{T}x_i-T \mu_X\geq \epsilon T \mu_X]\leq e^{-\frac{\epsilon^2}{2+\frac{2}{3}\epsilon}T \mu_X}
\end{align}
\end{lemma}

Here we consider Bernoulli random variables $X_{S_k}^i$ with mean $\mu_{X_{S_k}}=\frac{\mathbb{I}^l_T(S_k)}{|V^l|}$. 
	\begin{align*}\scriptsize
		&\Pr[|\mathcal{R}^l|\leq Q^l] \leq \Pr[\sum_{i=1}^{|\mathcal{R}^l|}X_{S_k}^i\leq \sum_{i=1}^{Q^l}X_{S_k}^i]\\
        &=\Pr[C_{\mathcal{R}^l}(S_k)\leq \sum_{i=1}^{Q_l}X_{S_k}^i]\leq Pr[\gamma^l\leq \sum_{i=1}^{Q_l}X_{S_k}^i]\\
		&\leq Pr[\frac{(1+\log \frac{3L^2}{(2L-1)\delta}\frac{1}{\phi^2})Q^l\times \mathbb{I}(S_k)}{|V^l|}\times \frac{\mathbb{I}(S_k^*)}{\mathbb{I}(S_k)}\leq \sum_{i=1}^{Q_l}X_{S_k}^i]\\
        &\leq  Pr[\sum_{i=1}^{Q^l}X_{S_k}^i-Q^l\mu_{X_{S_k}}\geq \log \frac{3L^2}{(2L-1)\delta}\frac{1}{\phi^2}Q^l\mu_{X_{S_k}}\frac{\mathbb{I}(S^*_k)}{\mathbb{I}(S_k)}] \\
        &\leq \exp(-\frac{(\log \frac{3L^2}{(2L-1)\delta}\frac{1}{\phi^2})^2(\frac{\mathbb{I}(S^*_k)}{\mathbb{I}(S_k)})^2}{2+\frac{2}{3}(\log \frac{3L^2}{(2L-1)\delta}\frac{1}{\phi^2})\frac{\mathbb{I}(S^*_k)}{\mathbb{I}(S_k)}}Q^l \mu_{X_{S_k}}) \quad\text{(Lemma \ref{lemma:martingale})}\\
        &\leq \exp(-\frac{(\log \frac{3L^2}{(2L-1)\delta}\frac{1}{\phi^2})^2\mathbb{I}(S^*_k)}{2\log \frac{3L^2}{(2L-1)\delta}\frac{1}{\phi^2}|V^l|}Q^l)\\
        &\leq \exp(-\frac{\log \frac{3L^2}{(2L-1)\delta}\frac{1}{\phi^2}\mathbb{I}(S^*_k)}{2|V^l|}Q^l)\leq \frac{2L-1}{3L^2}\delta\qquad\qquad\QED
\end{align*}
\end{proof}

\end{document}